\documentclass[12pt]{article}

\usepackage[top=1in,bottom=1in, left=1in, right=1in]{geometry} 
 
\usepackage{amsmath,amsthm,amssymb,amsopn,amsfonts,pdfpages,algorithmic,algorithm,dsfont}
\usepackage{graphics}
\usepackage{graphicx} 
\usepackage{bbm,subfig}
\usepackage{url}

\usepackage{authblk}
\usepackage[font=footnotesize]{caption} 
\usepackage[colorlinks=true,linkcolor=cyan,citecolor=blue]{hyperref}
\usepackage{blkarray}

\newcommand{\twospace}{\renewcommand{\baselinestretch}{1.3}\normalsize}

\newcommand{\Acnm}{\mathcal{A}^{\,d_1,e_1,d_2,e_2}_{n,m}}

\DeclareMathOperator*{\argmin}{arg\,min}
\DeclareMathOperator*{\argmax}{arg\,max}

\newcommand{\vnase}{\text{VN}\circ\text{GMM}\circ\text{ASE}}

\newcommand{\nnmf}{\mathcal{N}^{(n,m)}_{(d_1,e_1),(d_2,e_2)}}
\newcommand{\nnm}{\mathcal{N}^{(n,m)}}
\newcommand{\gnv}{\mathcal{G}_{n,\cv,\ce}}
\newcommand{\gmv}{\mathcal{G}_{m,\cv,\ce}}

\newcommand{\VNA}{\text{VN}\circ\text{GMM}\circ\text{ASE}}

\newcommand{\gn}{\mathcal{G}_n}
\newcommand{\gm}{\mathcal{G}_m}
\newcommand{\bbg}{ { \bf g}}

\newcommand{\gog}{({\bf g_1}, [\mo(\bbg_2)])}

\newcommand{\gmf}{\mathcal{G}_{m,\cv,\ce}^{\,d_2,e_2}}
\newcommand{\gnf}{\mathcal{G}_{n,\cv,\ce}^{\,d_1,e_1}}

\newcommand{\bH}{\mathbb{H}}

\newcommand{\calT}{{\mathcal T}}

\newcommand{\Z}{\mathbb{Z}}
\newcommand{\R}{\mathbb{R}}

\newcommand{\p}{\mathbb{P}}

\newcommand{\mre}{\widetilde{\mathcal{E}}}
\newcommand{\ce}{\mathcal{E}}
\newcommand{\cv}{\mathcal{V}}
\newcommand{\FF}{F^{(n,m)}_{c,\theta,(d_1,e_1),(d_2,e_2)}}
\newcommand{\FFt}{F^{(n,m)}_{\Theta}}

\newcommand{\e}{\mathbb{E}}

\newcommand{\mO}{\mathfrak{O}}

\newcommand{\rank}{\operatorname{rank}}

\newcommand{\id}{\operatorname{id}}

\newcommand{\rF}{{r_{\operatorname{F}}}}
\newcommand{\rGF}{{r_{\operatorname{GF}}}}
\newcommand{\rG}{{r_{\operatorname{G}}}}

\newcommand{\bW}{{\bf W}}

\newcommand{\bZ}{{\bf Z}}
\newcommand{\bX}{{\bf X}}
\newcommand{\bY}{{\bf Y}}
\newcommand{\bw}{{\bf w}}
\newcommand{\bz}{{\bf z}}
\newcommand{\Rb}{\mathbb{R}}
\newcommand{\bx}{{\bf x}}
\newcommand{\by}{{\bf y}}

\newcommand{\AS}{\mathfrak{A}_{n,m,\cv,\ce}^{(d_1,e_1,d_2,e_2)}}

\newcommand{\mo}{\mathfrak{o}}

\newtheorem{theorem}{Theorem}

\newtheorem{assu}[theorem]{Assumption}
\newtheorem{lemma}[theorem]{Lemma}

\theoremstyle{plain}
\newtheorem{definition}[theorem]{Definition}
\theoremstyle{plain}
\newtheorem{example}[theorem]{Example}
\theoremstyle{remark}
\newtheorem{remark}[theorem]{Remark}

\twospace

\begin{document}

\title{Vertex Nomination in Richly Attributed Networks}
\author[$1$]{Keith Levin}
\author[$2,3$]{Carey E. Priebe}
\author[$4$]{Vince~Lyzinski}

\affil[$1$]{\small Department of Statistics, University of Wisconsin-Madison}
\affil[$2$]{\small Department of Applied Mathematics and Statistics, Johns Hopkins University}
\affil[$3$]{\small Center for Imaging Sciences, Johns Hopkins University}
\affil[$4$]{\small Department of Mathematics, University of Maryland-College Park}

\maketitle

\begin{abstract}
Vertex nomination is a lightly-supervised network information retrieval task in which vertices of interest in one graph are used to query a second graph to discover vertices of interest in the second graph. Similar to other information retrieval tasks, the output of a vertex nomination scheme is a ranked list of the vertices in the second graph, with the heretofore unknown vertices of interest ideally concentrating at the top of the list.  Vertex nomination schemes provide a useful suite of tools for efficiently mining complex networks for pertinent information.
In this paper, we explore, both theoretically and practically, the dual roles of content (i.e., edge and vertex attributes) and context (i.e., network topology) in vertex nomination. We provide necessary and sufficient conditions under which vertex nomination schemes that leverage both content and context outperform schemes that leverage only content or context separately.  
While the joint utility of both content and context has been demonstrated empirically in the literature, the framework presented in this paper provides a novel theoretical basis for understanding the potential complementary roles of network features and topology.
\end{abstract}


\section{Introduction}
\label{sec:intro}

Network data has become ubiquitous in the sciences,
owing to the generality and flexibility
of networks in modeling relations among entities.
Networks appear in such varied fields as
neuroscience, genomics, the social sciences, economics and ecology,
to name just a few (see, for example, \cite{Newman2018}).
As such, statistical analysis of network data has
emerged as an important field within modern statistics \cite{kolaczyk2010statistical,kolaczyk2014statistical,crane2018probabilistic}.
Many classical statistical inference tasks,
such as 
hypothesis testing \cite{MT2,tang2014nonparametric,lei2014goodness,ginestet2017hypothesis}, 
regression \cite{fosdick2015testing,marrs2017standard},
and maximum likelihood estimation \cite{Bickel2009,snijders2010maximum,arroyo2018maximum} 
have been adapted to network data.
Inference tasks that are specific to network data,
such as
link-prediction \cite{liben2007link}, 
community detection \cite{NewCla2016,rohe2011spectral,SusTanFisPri2012}, and
vertex nomination \cite{marchette2011vertex,Coppersmith2014,suwan2015bayesian,FisLyzPaoChePri2015}
have also seen increasing popularity in recent years.
Among these network-specific tasks is the
{\em vertex nomination} (VN) problem,
in which the goal is to identify vertices similar to one or more
vertices specified as being of interest to a practitioner.
The VN task is
similar in spirit to popular network-based information retrieval (IR) procedures such as \texttt{PageRank} \cite{page1999pagerank} and 
personalized recommender systems on graphs \cite{huang2004graph}.
In VN, the goal is as follows:
Given vertices of interest in a graph $G_1$, produce a ranked list of the vertices in a second graph $G_2$ according to how likely they are judged to be interesting.
Ideally, interesting vertices in $G_2$ should concentrate
at the top of the ranked list.
As an inference task, this formulation of VN is distinguished from other supervised network IR tasks by the generality of what may define vertices as interesting and the limited available training data in $G_1$.
In contrast to typical IR problems, there is little or no training data available in the VN problem.

The vertex nomination problem was first introduced as a
task involving only a single graph,
and vertices of interest were modeled as belonging to a single
community of vertices
\cite{Coppersmith2014,FisLyzPaoChePri2015,lyzinski2016consistency,yoder2018vertex}.
The information provided by vertices with known community memberships,
called {\em seed vertices},
was leveraged to rank vertices with unknown membership,
with both network-topology and available vertex features being leveraged to produce ranking schemes \cite{marchette2011vertex,CopPri2012,suwan2015bayesian}.
This single-graph, community-based definition of the problem
is somewhat limited in its ability to capture network models
beyond the stochastic blockmodel \cite{Holland1983}.
Subsequent work lifted the problem to the two-network setting considered here \cite{patsolic2017vertex},
allowing a generalization of what defines interesting vertices and a generalization of the network models that could be considered \cite{patsolic2017vertex,lyzinski2017consistent,agterberg2019vertex}.

In many settings, observed networks are endowed with features at the vertex and/or edge level.
For example, in social networks, vertices typically correspond to users for whom we have demographic information, and edges correspond to different types of social relations.
The theoretical advances in both the single- and multiple-graph VN problem recounted above were established in the context of networks where no such feature are available.
It is natural, then, to seek to better understand the effect of
network attributes on the theoretical VN framework developed in \cite{lyzinski2017consistent} and \cite{agterberg2019vertex}.
Motivated by this, in the present work we develop VN on {\em richly-featured} networks, and we explore how the incorporation of this information impacts the concepts of Bayes optimality and consistency for the VN problem.
Furthermore, in Sections~\ref{sec:nofeat} and~\ref{sec:nogs},\ adopting an information theoretic perspective, we give the first steps toward a theoretical understanding (which is born out in subsequent experiments) of the potential benefit of VN schemata that use both content and context versus one of content or context alone.

The remainder of the paper is laid out as follows. 
In Section~\ref{sec:vnfeat}, we outline the extension of the vertex nomination framework to the attributed network setting, defining richly-featured graphs in Section~\ref{sec:RFN}, VN schema in Section~\ref{sec:VNschemes}, and VN performance measures in Section~\ref{sec:loss}.
In Section~\ref{sec:BOO}, we derive the Bayes optimal VN scheme in the setting of richly-featured networks, and in Sections~\ref{sec:nofeat} and~\ref{sec:nogs} we compare VN performance in the richly-featured setting to that in the
settings where feature information or network information,
respectively, is missing.
Experiments further illustrating the practical implications of our theoretical results are presented in Section \ref{sec:exp}.


\section{Vertex Nomination with Features}
\label{sec:vnfeat}
In the initial works introducing vertex nomination,
where the defining trait of interesting vertices was membership
in a community of interest, graph models with latent community structure (e.g., the stochastic blockmodel of \cite{Holland1983,karrer11:_stoch})
were sensible models for the underlying network structure.
The need for a more general notion of what renders a vertex interesting necessitated more nuanced models, culminating in the \emph{Nominatable Distribution} network model introduced in \cite{lyzinski2017consistent}.
We take this model as our starting point, and extend it by endowing it with both edge and vertex features.


\subsection{Richly Featured Networks}
\label{sec:RFN}

We begin by defining the class of networks with vertex and edge features,
which we call {\em richly-featured networks}.
We note here that there is a large literature on inference within attributed networks, with graphs endowed with features arising in settings such as social network analysis \cite{yl1,yl2} and knowledge representation \cite{kg1,kg2}, among others.

\begin{definition}\label{def:richfeat}
Let $\mathcal{V}$ and $\mathcal{E}$ be discrete sets of possible vertex and edge features, respectively.
A {\em richly-featured network} ${\bf g}$ indexed by $(n,d_1,d_2,\mathcal{V},\mathcal{E})$ is an ordered tuple 
${\bf g}=(g,\bx,\bw)$
where
\begin{itemize}
\item[i.] $g=(V,E)\in \mathcal{G}_n$ is a labeled, undirected graph on $n$ vertices.  The vertices of $g$ will be denoted via either $V=\{v_1,v_2,\cdots,v_n\}$ or $V=\{u_1,u_2,\cdots,u_n\}$.
\item[ii.] $\bx\in\mathcal{V}^{n\times d_1}$ denotes the matrix of $d_1$-dimensional vertex features, so that $\bx[v,:]$ is the vector of features associated with vertex $v$.
\item[iii.] Let $\mre=\mathcal{E}\cup\{\star\}$, where we use $\star$ as a special symbol representing unavailable data.
Letting $N=\binom{n}{2}$, $\bw\in\mre^{N\times d_2}$ denotes the matrix of $d_2$-dimensional edge features.
Indexing $\binom{V}{2}$ lexicographically, 
for $e\in\binom{V}{2}$, we write $\bw[e,:]$ for the vector of features associated with edge $e$.
The form of $\bw$ is then
$$\bw=\begin{pmatrix}
\bw[\{v_{1},v_{2}\},:]\\
\bw[\{v_{1},v_{3}\},:]\\
\bw[\{v_{1},v_{4}\},:]\\
\vdots\\
\bw[\{v_{n-1},v_{n}\},:]
\end{pmatrix}.$$
We further require that
$\bw[e,:]=(\star,\star,\cdots,\star)\in\mre^{d_2}$
for all $e\notin E$,
and $\bw[e,:] \in \mathcal{E}^{d_2}$ for all $e\in E$.
\end{itemize}
We use $\mathcal{G}_{n,\cv,\ce}^{\,d_1,d_2}$ to denote
the set of all richly-featured networks indexed by
$(n,d_1,d_2,\mathcal{V},\mathcal{E})$.
\end{definition}
\noindent Let $e\in\binom{V}{2}$.
In the definition of richly-featured networks, 
for $e\notin E$, we interpret the edge features $\bw[e,:]$ as unavailable data.  
This is a sensible assumption in practice, 
and is commonly made in attributed network models;  
see, for example \cite{pfeiffer2014attributed,zhou2012factorized}.
We note that the structure of $\bw$ encodes the edge structure of $g$,
but we choose to keep the redundant information in Definition~\ref{def:richfeat}, as $g$ encodes the purely topological structure of the graph, absent any edge- or vertex-level features.
This fact will prove useful below, when we seek investigate the role of graph structure in the absence of features and vice versa.

\begin{remark}
We use discrete vertex and edge feature sets in Definition~\ref{def:richfeat}, as this is both rich enough to model many real world networks (where features often encode types or characteristics of nodes or edges, and edge weights often derive from count data) and amenable to the theoretical derivations in vertex nomination.
Considering continuous features is not a practical problem, but does raise subtle measure-theoretic difficulties in the theory to follow.
See Remark~\ref{rem:measurability} for further discussion.
\end{remark}


\begin{example}
Consider the graph $\bbg\in\mathcal{G}_{4,\cv,\ce}^{\,d_1,d_2}$ with 
$g$ given by
\begin{equation*}
g=(V_g,E_g)=\left(\{1,2,3,4\},\{\ \{1,2\},\{1,3\},\{1,4\},\{3,4\}\ \} \right).
\end{equation*}
The edge features for this network would then be of the form
$$\bw=\begin{pmatrix}
\bw[\{v_1,v_2\},:]\\
\bw[\{v_1,v_3\},:]\\
\bw[\{v_1,v_4\},:]\\
\bw[\{v_2,v_3\},:]\\
\bw[\{v_2,v_4\},:]\\
\bw[\{v_3,v_4\},:]
\end{pmatrix}
=\begin{pmatrix}
\bw[\{v_1,v_2\},:]\in\mathcal{E}^{d_2}\\
\bw[\{v_1,v_3\},:]\in\mathcal{E}^{d_2}\\
\bw[\{v_1,v_4\},:]\in\mathcal{E}^{d_2}\\
\star,\star,\cdots,\star\\
\star,\star,\cdots,\star\\
\bw[\{v_3,v_4\},:]\in\mathcal{E}^{d_2}
\end{pmatrix}.$$
\end{example}
\begin{remark}
    Let $(n,d_1,d_2)$ be an ordered tuple of nonnegative integers, and let $\cv$ and $\ce$ be discrete sets of vertex and edge features, respectively.
    In the definitions and exposition that follow, we consider $\gnv^{d_1,e_1}$-valued random variables.  
    Implicitly, we mean the following: letting $(\Omega,\mathcal{F},\p)$ be a given probability space, 
    $(G,\bX,\bW):\Omega\mapsto\gnv^{d_1,d_2}$ is a $\gnv^{d_1,d_2}$-valued random variable if it is $(\mathcal{F},\mathcal{F}_{G_n}\otimes \mathcal{F}_{d_1}\otimes \mathcal{F}_{d_2}^*)$-measurable, where 
    $\mathcal{F}_{G_n}$ is the total sigma field on $\gn$, $\mathcal{F}_{d_1}$ is the total sigma field on $\cv^{n\times d_1}$, and 
    $\mathcal{F}_{d_2}^*$ is the total sigma field on $\mre^{N\times d_2}$.
\end{remark}

With Definition \ref{def:richfeat} in hand, lifting the definition of
 {\em Nominatable Distributions} first introduced in \cite{lyzinski2017consistent}
 to the attributed graph setting is relatively straightforward.
\begin{definition} \label{def:richfeatNom}
Given $n,m\in\Z_{>0}$ and sets of discrete vertex and edge features $\cv$ and $\ce$, respectively, the set of \emph{Richly Featured Nominatable Distributions of order $(n,m)$ with feature sets $\cv$ and $\ce$}, which we denote $\mathcal{F}_{\cv,\ce}^{(n,m)}$, is the collection of all families of distributions of the form 
\begin{equation*} \begin{aligned}
\mathbf{F}^{(n,m)}=\Big\{&\FF \ : \ 
\big(c,\theta,(d_1,e_1),(d_2,e_2)\big) \in \Z_{\geq 0}\times\mathbb{R}^{d(n,m)}\times \Z^2_{>0}\times \Z^2_{>0}, \\
&~~~~~~~~~~~~\text{ and } 0\leq c\leq\min(n,m) \Big\},
\end{aligned} \end{equation*}
where $\FF$ is a distribution on $\mathcal{G}_n \times \cv^{n\times d_1}\times \mre^{N\times e_1}\times\mathcal{G}_m\times \cv^{m\times d_2}\times \mre^{M\times e_2}$
(recalling that $N=\binom{n}{2}, M=\binom{m}{2}$),
parameterized by $\theta \in \R^{d(n,m)}$
and satisfying the following conditions: 
\begin{enumerate}
\item{The vertex sets $V_1 = \{v_1,v_2,...,v_n\}$ and 
$V_2 = \{u_1,u_2,...,u_m\}$ satisfy $v_i = u_i$ for $1 \le i \leq c$. 
We refer to $C = \{v_1,v_2,...,v_c\} = \{u_1,u_2,...,u_c\}$ as the {\em core vertices}. These are the vertices that are shared across the two graphs and imbue the model with a natural vertex correspondence.}
\item{Vertices in $J_1 = V_1 \setminus C$ and $J_2 = V_2 \setminus C$, satisfy $J_1 \cap J_2 = \emptyset$. We refer to $J_1$ and $J_2$ as {\em junk vertices}. These are the vertices in each graph that have no corresponding vertex in the other graph}
\item{
If $(G_1,\bX,\bW,G_2,\bY,\bZ)$ is distributed according to $\FF$, then
$(G_1,\bX,\bW)$ is a $\gnv^{\,d_1,e_1}$-valued random variable and  
$(G_2,\bY,\bZ)$ is a $\gmv^{\,d_2,e_2}$-valued random variable.
The edge features $\bW\in\mre^{N\times e_1}$ and $\bZ\in\mre^{M\times e_2}$ almost surely satisfy
$$\bW[e,:]=\begin{cases}
(\star,\star,\cdots,\star)\in\mre^{e_1}&\text{ if }e\notin E(G_1);\\
\in \ce^{e_1}&\text{ if }e\in E(G_1);
\end{cases}$$
and 
$$\bZ[e,:]=\begin{cases}
(\star,\star,\cdots,\star)\in\mre^{e_2}&\text{ if }e\notin E(G_2);\\
\in \ce^{e_2}&\text{ if }e\in E(G_2).
\end{cases}$$
}
\item{The richly-featured subgraphs induced by the junk vertices,
$$\left(G_1[J_1],{\bf X}[J_1,:],{\bf W}\left[\binom{J_1}{2},\,:\,\right]\right)
	\text{ and }
\left(G_2[J_2],{\bf Y}[J_2,:],{\bf Z}\left[\binom{J_2}{2},\,:\,\right]\right)$$
are conditionally independent given $\theta$.}
\end{enumerate}
\end{definition}

\noindent In Definition~\ref{def:richfeatNom}, the rows of
${\bf X}\in \cv^{n\times d_1}$ are the vertex features of $G_1$, with 
$\bX[i,:]$ representing the feature associated with vertex $i$ in $G_1$.
Similarly, the rows of
${\bf Y}\in \cv^{n\times d_2}$ are the vertex features of $G_2$,
with 
$\bY[i,:]$ representing the vertex feature of vertex $i$ in $G_2$.
We do not, a priori, assume that any vertex features are missing, although extending the definition to $\widetilde\cv=\cv\cup\{\star\}$ is straightforward.
With this definition in place, we are ready to define feature-aware vertex nomination schemes.
\vspace{2mm}

\noindent{\bf Note:} In order to ease notation below, we will write
\begin{equation*}
    \Theta:=(c,\theta,(d_1,e_1),(d_2,e_2)),
\end{equation*}
and accordingly write $\FFt$ for $\FF$.
In the sequel, we will assume that the feature sets $\ce$ and $\cv$ are given, and satisfy $|\ce|=|\cv|=\infty$.
We will suppress the dependence of the family of richly-featured nominatable distributions on the feature sets $\ce$ and $\cv$, writing $\mathcal{F}^{(n,m)}$ in place of $\mathcal{F}_{\cv,\ce}^{(n,m)}$.

\subsection{Vertex Nomination Schemes}
\label{sec:VNschemes}

In vertex nomination, the labels of vertices in the second graph, $\bbg_2$, are assumed unknown a priori.  
In order to accomplish this in our Featured Nominatable Distribution framework, we introduce
{\em obfuscating functions} as in \cite{lyzinski2017consistent}.
Obfuscation functions serve to hide vertex labels, and can be interpreted as a non-probabilistic version of the vertex shuffling considered in \cite{vogelstein2011shuffled} and \cite{lyzinski2016information}.

\begin{definition} \label{def:obfuscate}
Consider graphs $(\bbg_1,\bbg_2)\in\gnv^{d_1,e_1}\times\gmv^{d_2,e_2}$ with vertex sets $V_1$ and $V_2$, respectively.
An {\em obfuscating set}, $H$, of $V_1$ and $V_2$ of order $|V_2|=m$ is a set satisfying $H \cap V_i = \emptyset$ for $i = 1,2,$ and $|H|=|V_2|=m$.
Given obfuscating set $H$, an \emph{obfuscating function} $\mo : V_2 \rightarrow H$ is a bijection from $V_2$ to $H$. We denote by $\mO_H$ the set of all such obfuscating functions.
For a richly-featured network ${\bf g}=(g,\bx,\bw)\in \gmv^{\,d_2,e_2}$, we will write 
$\mo({\bf g})=(\mo(g),\mo(\bx),\mo(\bw))$ where
\begin{itemize}
\item[i.] $\mo(g)$ denotes the graph $g = (V_g,E_g)$ with labels obfuscated by $\mo$. That is, $\mo(g)=(V_{\mo(g)},E_{\mo(g)})$, where $V_{\mo(g)}=\{ \mo(v) : v \in V_g)$ and $E_{\mo(g)}$ is such that $\{u,v\}\in E_g$ if and only if $\{\mo(u),\mo(v)\}\in E_{\mo(g)}$.
\item[ii.] $\mo(\bx)$ is the vertex feature matrix associated with $\mo(g)$, so that for $u\in H$,
$$
(\mo(\bx))[u,:]=\bx[\mo^{-1}(u),:].
$$
\item[iii.] $\mo(\bw)$ is the edge feature matrix associated with $\mo(g)$, so that for $\{v,u\}\in \binom{H}{2}$,
$$
\left(\mo(\bw)\right)[\{v,u\},:]=\bw\left[\left\{\mo^{-1}(v),\mo^{-1}(u)\right\},:\right].
$$
\end{itemize}
Note that we will assume that $H$ is endowed with an arbitrary but fixed ordering, and that the edges of $\mo(\bw)$ are ordered lexicographically according to this ordering on $H$.
We do not necessarily assume that the ordering of $H$ is the ordering induced by $V$. That is, we do not necessarily assume that $u\leq v$ implies $\mo(u)\leq\mo(v)$.
\end{definition}

Relating this definition back to Definition~\ref{def:richfeatNom}, the purpose of the obfuscating function is to render the labels on the vertices in $\bbg_2$ uninformative with respect to the correspondence between $\bbg_1$ and $\bbg_2$ encoded in the core vertices $C$.
Following this logic, it is sensible to require vertex nomination schemes (defined below) to be independent of vertex labels.
Informally, if a set of vertices have identical features and edge structures among them, then their rankings in a VN scheme should be independent of the chosen obfuscating function $\mo\in\mO_H$.
This is made precise in Definition~\ref{def:VN} and Assumption~\ref{ass:FAVN} below, but requires some preliminary definitions.

\begin{definition}[Action of a permutation on a richly-featured network]
Let ${\bf g}=(g,\bx,\bw)\in \gnv^{\,d_1,d_2}$ be a richly-featured network.
A permutation $\sigma:[n]\mapsto[n]$ acts on ${\bf g}$ to produce
$\sigma({\bf g}) =(g',\bx',\bw') \in \gnv^{\,d_1,d_2}$,
where
\begin{itemize}
%
\item[i.] $g'=\sigma(g)$ is the graph $g$ with its vertex labels permuted by $\sigma$.
\item[ii.] $\bx'$ is the vertex feature matrix associated with $g'$, so that for $v\in[n]$,
$$
\bx'[v,:]=\bx[\sigma^{-1}(v),:].
$$
\item[iii.] $\bw'$ is the edge feature matrix associated with $g'$, so that for $\{u,v\}\in \binom{[n]}{2}$,
$$
\bw'[\{u,v\},:]=\bw\left[\left\{\sigma^{-1}(u),\sigma^{-1}(v)\right\},:\right].
$$
\end{itemize}
\end{definition}

\begin{definition}[feature-preserving automorphisms and isomorphisms]
\label{def:fauto}
We call a permutation $\sigma$ a \emph{feature-preserving automorphism} (abbreviated f-automorphism) of ${\bf g}$ if ${\bf g}=\sigma({\bf g})$.
Similarly, We call a permutation $\sigma$ a \emph{feature-preserving isomorphism} between ${\bf g}$ and ${\bf g}'$ (abbreviated f-isomorphism) if ${\bf g}'=\sigma({\bf g})$.
\end{definition}

Let ${\bf g}=(g,\bx,\bw)\in \gnv^{\,d_1,d_2}$ be a richly-featured network.
For each $u\in V_g$, define
\begin{equation*}
\mathcal{I}(u;{\bf g}):=\{w\in V_g\text{ s.t. } 
	\exists\text{ an f-automorphism }\sigma\text{ of }{\bf g},\text{ s.t. }\,\sigma(u)=w\}.
\end{equation*}
With the above notation in hand, we are now ready to introduce the concept of a feature-aware vertex nomination scheme (often abbreviated {\em VN scheme} in the sequel).
In the definition to follow, $V^*$ represents the set of \emph{vertices of interest} in $\bbg_1$.
These are usually assumed to be in $V_1\cap V_2$, and the goal of a VN scheme is to have $\mo(V^*)$ concentrate at the top of the produced rank list in $\mathcal{T}_H$.
\begin{definition}[Feature-aware VN Scheme]
\label{def:VN}
Let $n,m,d_1,e_1,d_2,e_2\in \Z_{>0}$ and $\cv$, $\ce$ be given.
Let $H$ be an obfuscating set of $V_1$ and $V_2$ of order $|V_2|=m$, and let $\mo\in\mathfrak{O}_H$ be given.
For a set $A$, let $\calT_A$ denote the set of all total orderings of the elements of $A$.
A {\em feature-aware vertex nomination scheme (FA-VN scheme)} on $\gnv^{\,d_1,e_1} \times \mo(\gmv^{\,d_2,e_2})$ is a function 
\begin{equation*}
\Phi: \gnv^{\,d_1,e_1} \times \mo(\gmv^{\,d_2,e_2}) \times 2^{V_1} \rightarrow \calT_{H}
\end{equation*}
satisfying the consistency criteria in Assumption~\ref{ass:FAVN}.
We let $\mathcal{N}^{(n,m)}=\mathcal{N}^{(n,m)}_{(d_1,e_1),(d_2,e_2)}$ denote the set of all such VN schemes.
\end{definition}

The consistency criteria required of FA-VN schemes essentially forces the schemes to be agnostic to the labels in the obfuscated $\mo({\bf g}_2)$.  
To accomplish this, we define the following.
\begin{assu}[FA-VN Consistency criteria]
\label{ass:FAVN}
With notation as in Definition \ref{def:VN},
for each $u\in V_2$ and $V^* \subseteq V_1$, define
\begin{equation*}
\rank_{\Phi({\bf g_1 },\mo({\bf g_2}),V^*)}\big(\mo(u)\big)
\end{equation*}
to be the position of $\mo(u)$ in the total ordering provided by $\Phi({\bf g_1},\mo({\bf g_2}),V^*)$. 
Further, define
$\mathfrak{r}_{\Phi}:\gn^{\,d_1,e_1}\times\gm^{\,d_2,e_2}\times\mathfrak{O}_H\times 2^{V_1}\times2^{V_2}\mapsto 2^{[m]}$ according to
\begin{equation*}
\mathfrak{r}_{\Phi}({\bf g_1},{\bf g_2},\mo,V^*,S)=\{\rank_{\Phi({\bf g_1},\mo({\bf g_2}),V^*)}\big(\mo(u)\big)\text{ s.t. }u\in S \}.
\end{equation*}
For any ${\bf g_1}\in\gnv^{\,d_1,e_1},$ ${\bf g_2}\in\gmv^{\,d_2,e_2}$, $V^*\subset V_1$, obfuscating functions $\mo_1,\mo_2\in\mathfrak{O}_H$ and any $u\in V(g_2)$, we require
\begin{align}
\label{eq:consis}
&\mathfrak{r}_{\Phi}({\bf g_1},{\bf g_2},\mo_1,V^*,\mathcal{I}(u;{\bf g_2}))=\mathfrak{r}_{\Phi}\left({\bf g_1},{\bf g_2},\mo_2,V^*,\mathcal{I}(u;{\bf g_2}) \right)\\
\notag &\Leftrightarrow \mo_2\circ\mo_1^{-1}\big( \mathcal{I}(\Phi({\bf g_1},\mo_1({\bf g_2}),V^*)[k]);\mo_1({\bf g_2})\big)=\mathcal{I}\left( \Phi({\bf g_1},\mo_2({\bf g_2}),V^*)[k];\mo_2({\bf g_2}) \right)\\
& \hspace{25mm}\text{ for all }k\in[m],\notag
\end{align}
where $\Phi({\bf g_1},\mo({\bf g_2}),V^*)[k]$ denotes the $k$-th element (i.e., the rank-$k$ vertex) in the ordering $\Phi({\bf g_1},\mo({\bf g_2}),V^*)$.
\end{assu}

Figure~\ref{fig:goodvn} gives a simple illustrative example of this consistency criterion (i.e., Eq.~\ref{eq:consis}) in action.
Note here that if $\mathcal{I}(u;{\bf g}_2)=\{u\}$ for all $u\in V_2$, then the consistency criterion forces 
$$\Phi({\bf g_1},\sigma(\mo({\bf g_2})),V^*)=
\sigma(\Phi({\bf g_1},\mo({\bf g_2}),V^*))$$
for any permutation $\sigma$ and obfuscating $\mo\in\mO_H$.

\begin{figure}[ht!]
    \centering
    \begin{tabular}{c|c}
    \includegraphics[width=0.45\textwidth]{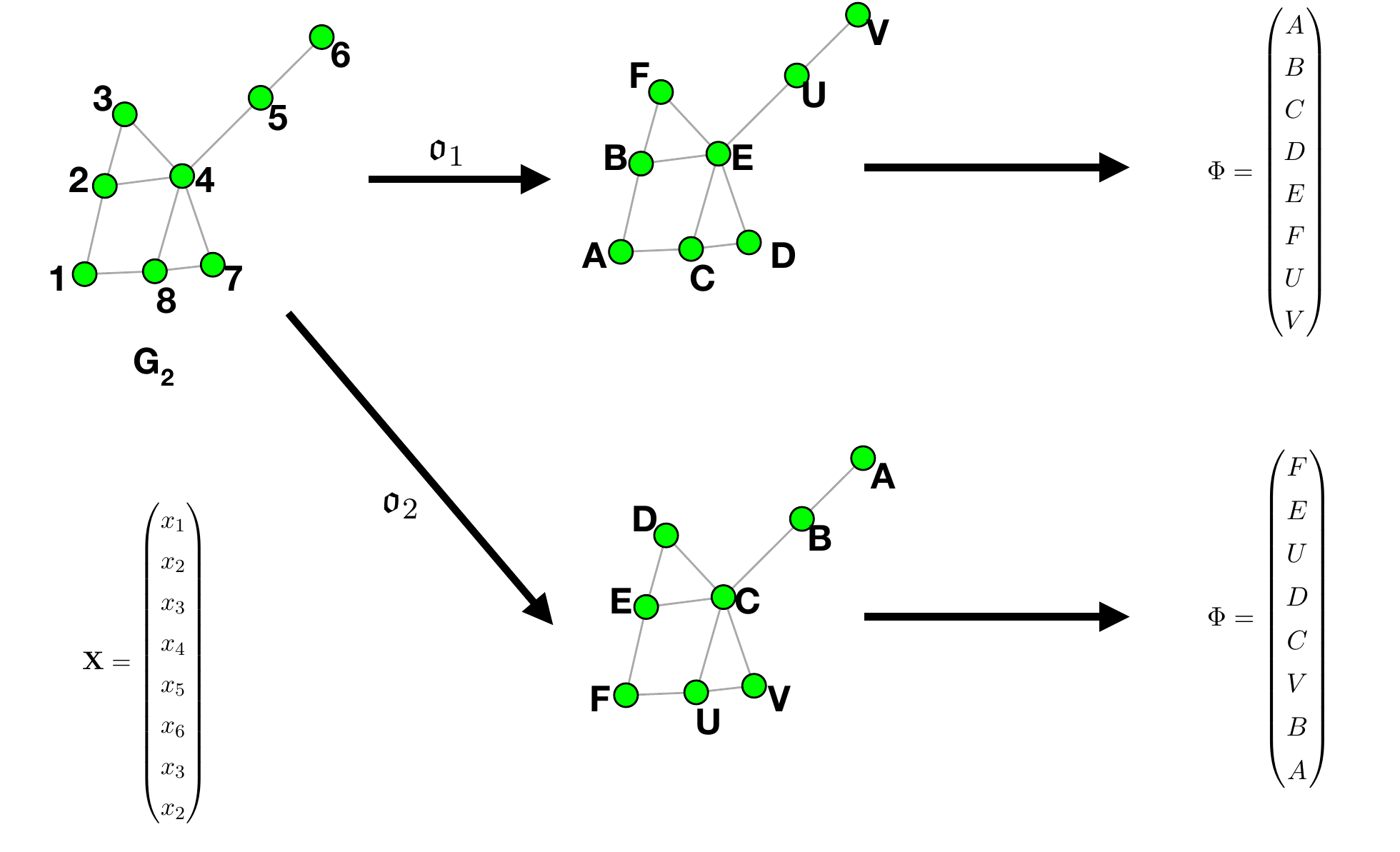} & 
     \includegraphics[width=0.45\textwidth]{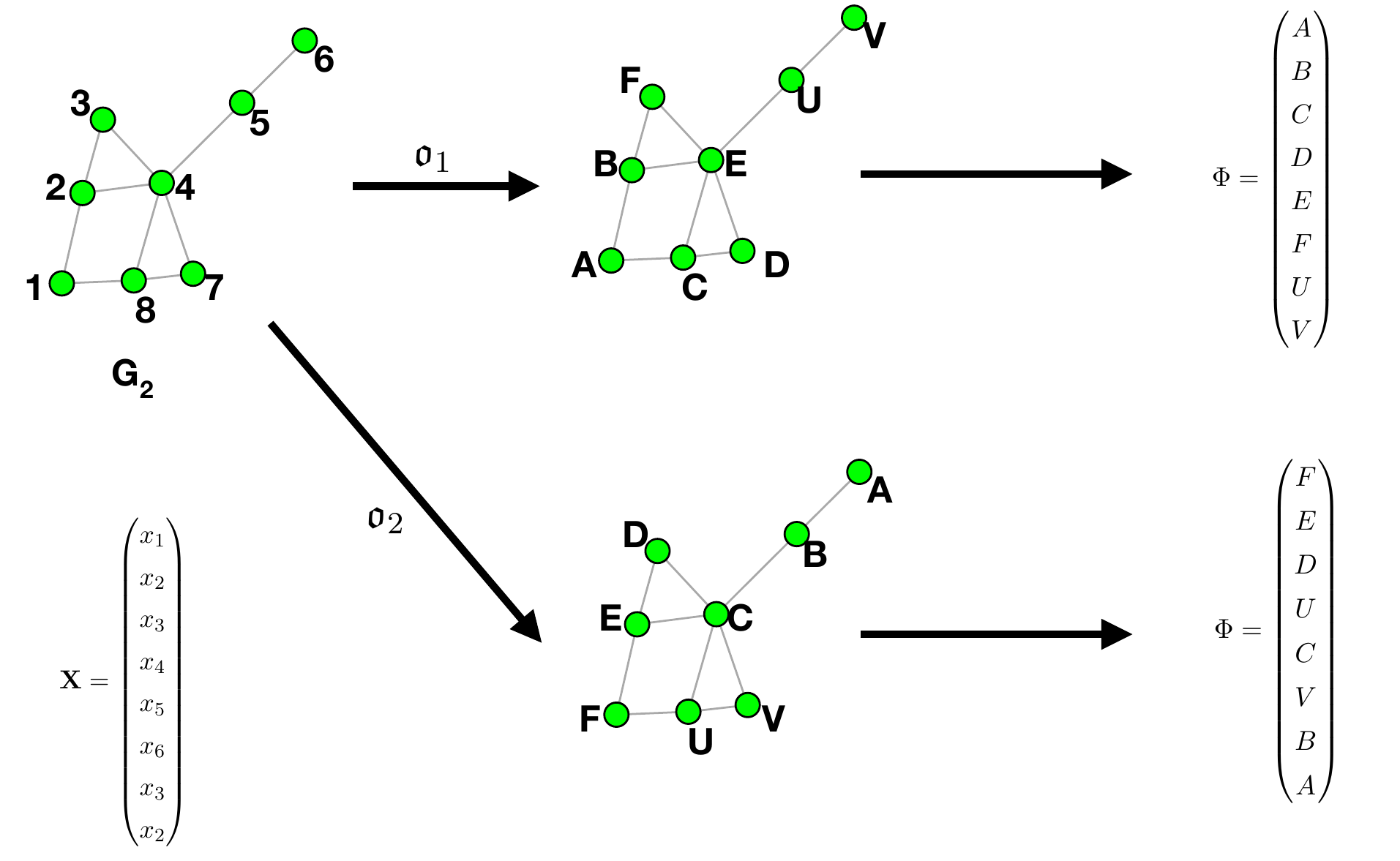}\\
     (a) Internally consistent scheme &
    (b)  Inconsistent scheme
    \end{tabular}
    \caption{
  Example of VN schemes that (a) satisfy and (b) do not satisfy the consistency criterion in Eq.\@ \ref{eq:consis}.
  Both subplots (a) and (b) denote two different obfuscation functions applied to the same graph.
  Here,
  $\mathcal{I}(1;G_2)=\{1\}$,
   $\mathcal{I}(2;G_2)=\mathcal{I}(8;G_2)=\{2,8\}$,
    $\mathcal{I}(3;G_2)=\mathcal{I}(7;G_2)=\{3,7\}$,
    $\mathcal{I}(4;G_2)=\{4\}$,
    $\mathcal{I}(5;G_2)=\{5\}$
    and $\mathcal{I}(6;G_2)=\{6\}$.
  Under $\mo_1$, $\{2,8\}$ is mapped to $\{B,C\}$, while under $\mo_2$, $\{2,8\}$ is mapped to $\{E,U\}$.
  The consistency property in Definition~\ref{def:VN} requires that the ranking of $\{2,8\}$ must be the same under $\mo_1$ and $\mo_2$.
  That is, $\{B,C\}$ and $\{E,U\}$ should be ranked the same under $\mo_1$ and $\mo_2$, respectively.
  This requirement is obeyed by the VN scheme illustrated in subplot (a).
  On the other hand, the scheme in subplot (b) violates the consistency property: it ranks $\{B,C\}$ as second and third under $\mo_1$, but ranks $\{E,U\}$ as second and fourth under $\mo_2$.
  }
    \label{fig:goodvn}
\end{figure}

In VN and other IR ranking problems, ties due to identical structure (here represented by f-isomorphisms in ${\bf g}_1$ or ${\bf g}_2$)
cause theoretical complications.
We refer the interested reader to \cite{lyzinski2017consistent}
and \cite{agterberg2019vertex} for examples of these complications and how they can be handled.
In order to avoid the additional notational and definitional burdens required to deal with tie-breaking in these situations,
we will make the following assumption on the distributions considered in $\mathcal{F}^{(n,m)}$.
\begin{assu}
\label{ass:A1}
Let $({\bf G}_1,{\bf G}_2)\sim F^{(n,m)}_{\Theta}\in\mathcal{F}^{(n,m)}$
and consider the events
\begin{align*}
D_1&=\{\text{ the only f-automorphism of }{\bf G}_1\text{ is }\sigma=\id_n  \}\\
D_2&=\{\text{ the only f-automorphism of }{\bf G}_2\text{ is }\sigma=\id_m  \}.
\end{align*} 
$F^{(n,m)}_{\Theta}$ satisfies $\p_{\FFt}( D_1)=\p_{\FFt}( D_2)=1$.
\end{assu}
This assumption is unrealistic if there are only a few categorical vertex features (for example, roles in a corporate hierarchy), but this assumption is less restrictive when there are a large number of available categorical features or the features are continuous.
We stress that this assumption is made purely to ease the presentation of theoretical material, and the practical impact of this assumption being violated is easily overcome.

\subsection{Loss and Bayes Loss}
\label{sec:loss}

A vertex nomination scheme is, essentially, a semi-supervised IR system for querying large networks.
Similar to the recommender system framework \cite{rshandbook}, a VN scheme is judged to be successful if the top of the nomination list contains a high concentration of vertices of interest from the second network.
This motivates the definition of VN loss based on the concept of precision-at-$k$. 

\begin{definition} \label{def:vnloss}
Let $\Phi \in \nnm=\nnmf$ be a VN scheme, $H$ an obfuscating set of $V_1$ and $V_2$ of order $|V_2|=m$, and $\mo\in\mathfrak{O}_H$.  
Let $({\bf g_1},{\bf g_2})$ be realized from 
\begin{equation*}
(G_1,\bX,\bW,G_2,\bY,\bZ) \sim \FFt\in\mathcal{F}^{(n,m)},
\end{equation*}
with a vertex of interest set $V^* \subset C=V_1\cap V_2$.
For $k \in [m-1]$, we define
\begin{itemize}
    \item[(i)] For $(\bbg_1,\bbg_2)$ realized as $(G_1,\bX,\bW,G_2,\bY,\bZ)$, the \emph{level-$k$ nomination loss}
\begin{equation*} \begin{aligned}
\ell_k (\Phi,{\bf g_1},\mo({\bf g_2}),V^*) :&=
1 - \frac{1}{k}\sum_{v \in V^*} \mathds{1}\{ \rank_{\Phi({\bf g_1}, \mo({\bf g_2}), V^*)}( \mo(v)) \leq k \} 
\end{aligned} \end{equation*}
\item[(ii)] The \emph{level-$k$ error} of $\Phi$ is defined as 
\begin{equation*} \begin{aligned}
L_k (\Phi, V^*) &=L_k (\Phi, V^*,\mo) := 
\e_{({\bf G_1},{\bf G_2}) \sim \FFt}[ \ell_k (\Phi, {\bf G_1},\mo({\bf G_2}), V^*)] \\
&=1- \frac{1}{k} \sum_{v \in V^*} 
\p_{\FFt} 
\bigg( \rank_{\Phi({\bf G_1}, \mo({\bf G_2}),V^*)}(\mo(v)) \leq k  \bigg),
\end{aligned} \end{equation*}
where ${\bf G_1}=(G_1,\bX,\bW)$ and ${\bf G_2}=(G_2,\bY,\bZ)$.
\end{itemize}
The \emph{level-k Bayes optimal scheme} for $\FFt$ is defined as any element
\begin{equation*}
\Phi^*_{k}=\Phi^*_{k,V^*}\in \argmin_{\phi\in\nnm} L_k (\Phi, V^*),
\end{equation*}
with corresponding Bayes error $L^{*}_k$.
\end{definition}

\begin{remark}
{Note that we could have also defined a recall-based loss function via
\begin{equation*}
\ell^{(r)}_k (\Phi,{\bf g_1},{\bf g_2},V^*) :=
\frac{1}{|V^*|}
\sum_{v \in V^*} \mathds{1}\{ \rank_{\Phi({\bf g_1}, \mo({\bf g_2}), V^*)}( \mo(v)) \geq k + 1\} .
\end{equation*}
We focus on the more natural precision-based loss function, but we note in passing that consistency and Bayes optimality with respect to these two loss functions is equivalent when $|V^*|=O(1)$.}
\end{remark}

\section{Bayes Optimal Vertex Nomination Schemes}
\label{sec:BOO}

In \cite{lyzinski2017consistent}, a Bayes optimal VN scheme (i.e., one that achieves optimal expected VN loss) was derived in the setting where one observes a network {\em without} features.
In the feature-rich setting, derivations are similar, though they require a more careful technical treatment. After some preliminary work, this section culminates in the definition of the feature-aware Bayes optimal scheme in Section~\ref{sec:BO}.


\subsection{Obfuscating Features}

We are now faced with the problem of modeling the effect of the obfuscating function on features under the VN framework.
If we observe $\mo(\bbg_2)$, then
we have no knowledge of which member of
\begin{equation*}
    [\bbg_2]:=\{\bbg_2':\bbg_2\simeq\bbg_2'\}
\end{equation*}
was obfuscated,
but we do know what features are associated to each of the vertices
and edges.
In order to model this setting, we adopt the following conventions.
Let
$n,m,d_1,e_1,d_2,e_2\in\mathbb{Z}_{>0}$
and $\cv$, $\ce$ be given.
Furthermore, let the set of vertices of interest, $V^*\subseteq C = V_1\cap V_2$, be fixed.
Let $H$ be an obfuscating set of of $V_1$ and $V_2$ of order $|V_2|=m$, and $\mo\in\mathfrak{O}_H$.
Define $\Acnm$ to be the set of asymmetric richly-featured graphs
\begin{align*}
\mathcal{A}_{n,m}=\Acnm
&:=\Big\{({\bf g}_1,{\bf g}_2)\in\gnf\times\gmf ~:~ \text{there are no }\\
&~~~~~~~~~\text{ non-trivial f-automorphisms of }{\bf g}_1\text{ or }{\bf g}_2\Big\}.
\end{align*}
Under Assumption \ref{ass:A1}, $\FFt$ is supported on $\mathcal{A}_{n,m}$.

For each $({\bf g_1},{\bf g_2})\in \gnv^{(d_1,e_1)}\times\gmv^{(d_2,e_2)}$, define
\begin{equation*}
    \begin{aligned}
\left({\bf g_1}, [\mo(\bbg_2)] \right) &= \bigg\{ (\bbg_1, \widehat \bbg_2) \in \gnv^{(d_1,e_1)}\times\gmv^{(d_2,e_2)}\text{ s.t. } \mo(\widehat \bbg_2) \simeq \mo(\bbg_2) \bigg\} \\
&=\bigg\{ (\bbg_1, \widehat \bbg_2) \in \gnv^{(d_1,e_1)}\times\gmv^{(d_2,e_2)}\text{ s.t. } \widehat \bbg_2 \simeq \bbg_2 \bigg\}.
\end{aligned}
\end{equation*}
Note that if $({\bf g_1},{\bf g_2})\in \mathcal{A}_{n,m}$ the asymmetry of ${\bf g_2}$ yields that 
\begin{equation*}
\Big|\Big({\bf g_1}, \big[\mo(\bbg_2)\big]\Big)\Big|=m!.
\end{equation*}
In light of the action of the obfuscating function on the features and vertex labels of $\bbg_2$, we view 
$(\bbg_1,[\mo(\bbg_2)])$ as the set of possible graph pairs that could have led to the observed graph pair 
$(\bbg_1,\mo(\bbg_2))$.

For each $u \in H$ and $v \in V_2$, we also define the following restriction:
\begin{equation*} \begin{aligned}
\gog_{u = \mo(v)} 
=\bigg\{& (\bbg_1, \widehat \bbg_2) \in \gnv^{(d_1,e_1)}\times\gmv^{(d_2,e_2)}\text{ s.t. } \mo(\widehat \bbg_2) = \sigma( \mo(\bbg_2)),\\
&~~~~~~\text{where $\sigma$ is an f-isomorphism satisfying }\sigma(u)=\mo(v) \bigg\} \\
=\bigg\{&  (\bbg_1, \widehat \bbg_2) \in \gnv^{(d_1,e_1)}\times\gmv^{(d_2,e_2)}\text{ s.t. } \widehat \bbg_2 = \sigma(\bbg_2),\\
&~~~~~~\text{ where $\sigma$ is an f-isomorphism satisfying } \sigma(\mo^{-1}(u))= v\bigg\},
\end{aligned} \end{equation*}
and for $S\subseteq V_2,$ define 
$\gog_{u \in \mo(S)}=\bigcup_{v\in S}\,\,\gog_{u = \mo(v)}.$

\subsection{Defining Bayes Optimality} \label{sec:BO}
We are now ready to define a Bayes optimal scheme $\Phi^*$
for a given $\FFt\in\mathcal{F}^{(n,m)}$ satisfying Assumption \ref{ass:A1}.
We will define the scheme element-wise on each asymmetric $(\bbg_1,[\mo(\bbg_2)])$, and then systematically lift the scheme to all of $\Acnm$.
To wit, let
\begin{equation*}
    \left\{(\bbg_1^{(i)},\bbg_2^{(i)}) : i \in I \right\}
\end{equation*}
form a partition of $\gnv^{d_1,e_1}\times \gmv^{d_1,e_1}$.
To ease notation, we adopt the following shorthand.
\begin{enumerate}
    \item We use $\big( \bbg^{(i)}_1,\left[\mo(\bbg^{(i)}_2) \right]\big)$
    to denote the event $\left\{({\bf G_1},\mo({\bf G_2}))\in\left( \bbg^{(i)}_1,\left[\mo(\bbg^{(i)}_2) \right]\right)\right\}$.
\item For $u\in H$, we use $\big( \bbg^{(i)}_1,\left[\mo(\bbg^{(i)}_2) \right] \big)_{u=\mo(v)}$
    to denote the event
    \begin{equation*}
    \left\{({\bf G_1},\mo({\bf G_2}))\in
    \left( \bbg^{(i)}_1,\left[\mo(\bbg^{(i)}_2) \right]\right)_{u=\mo(v)}\right\}.
    \end{equation*}
    We will use this often with $u=\Phi(\bbg^{(i)}_1,\mo(\bbg^{(i)}_2),V^*)[j]$.
    
\item We use $\big( \bbg^{(i)}_1,\left[\mo(\bbg^{(i)}_2) \right]\big)_
{u\in \mo(V^*)}$ to denote the event
\begin{equation*}
    \left\{({\bf G_1},\mo({\bf G_2}))\in
\left( \bbg^{(i)}_1,\left[\mo(\bbg^{(i)}_2) \right]\right)_
{u\in\mo(V^*)}\right\}.
\end{equation*}
    We will use this often with $u=\Phi(\bbg^{(i)}_1,\mo(\bbg^{(i)}_2),V^*)[j]$.
\end{enumerate}
Let $\mathcal{S}$ denote the set of indices $i$ such that $(\bbg_1^{(i)},\bbg_2^{(i)})\in\mathcal{A}_{n,m}$.
That is, $\bbg_1^{(i)}$ and $\bbg_2^{(i)}$ are asymmetric as richly-featured networks.
For each $i\in\mathcal{S}$, writing $\p(\cdot)$ for $\p_{\FFt}(\cdot)$ to ease notation,
define
\begin{equation} \label{eq:BOFA}
\begin{aligned}
\Phi^*(\bbg^{(i)}_1 , \mo(\bbg^{(i)}_2), V^*)[1] 
&\in \argmax_{\substack{u\in H }}   \,\,
\p\bigg( (\bbg^{(i)}_1, [\mo(\bbg^{(i)}_2)])_{u\in \mo(V^*) }\,\, \big|\,\, (\bbg^{(i)}_1, [\mo(\bbg^{(i)}_2)]) \bigg) \notag\\
\Phi^*(\bbg^{(i)}_1 , \mo(\bbg^{(i)}_2), V^*)[2] 
&\in \argmax_{\substack{u\in H\setminus\{ \Phi^*[1]\}}}  
\,\,
\p\bigg( (\bbg^{(i)}_1, [\mo(\bbg^{(i)}_2)])_{u\in \mo(V^*) }\,\, \big|\,\, (\bbg^{(i)}_1, [\mo(\bbg^{(i)}_2)]) \bigg) \notag \notag\\
&\vdots \notag\\
\Phi^*(\bbg^{(i)}_1 , \mo(\bbg^{(i)}_2), V^*)[m] 
&\in \argmax_{\substack{u\in H\setminus\{\cup_{j< m}\{\Phi^*[j]\}}}  
\,\,
\p\bigg( (\bbg^{(i)}_1, [\mo(\bbg^{(i)}_2)])_{u\in \mo(V^*) }\,\, \big|\,\, (\bbg^{(i)}_1, [\mo(\bbg^{(i)}_2)]) \bigg) \notag ,
\end{aligned}
\end{equation}
where we break ties in an arbitrary but fixed manner.
For each element 
\begin{equation*}
    (\bbg_1,\bbg_2)\in \left( \bbg^{(i)}_1, [\mo(\bbg^{(i)}_2)] \right),
\end{equation*}
choose the f-isomorphism $\sigma$ such that $\mo(\bbg_2)=\sigma(\mo(\bbg^{(i)}_2))$, and define
\begin{equation*} 
 {\Phi^*}(\bbg_1,\mo(\bbg_2),V^*)=\sigma({\Phi^*}(\bbg^{(i)}_1,\mo(\bbg^{(i)}_2),V^*)).
\end{equation*}
For $i\not\in\mathcal{S}$, any fixed and arbitrary definition of $\Phi^*$ (subject to the consistency criterion in Definition~\ref{def:VN}) on $(\bbg^{(i)}_1, [\mo(\bbg^{(i)}_2)])$ will suffice, as this set has measure $0$ under $\FFt$ by Assumption \ref{ass:A1}.
Theorem~\ref{thm:BO} shows that $\Phi^*$ defined above is indeed level-$k$ Bayes optimal for all $k$ for any nominatable distribution satisfying Assumption \ref{ass:A1}.
A proof is given in Appendix~\ref{sec:proof!}.
\begin{theorem}
\label{thm:BO}
Let $\FFt\in\mathcal{F}^{(n,m)}$ satisfy Assumption~\ref{ass:A1}, and let $V^*\subset C=V_1\cap V_2$ be a given set of vertices of interest in ${\bf G}_1$.
The FA-VN scheme
$\Phi^*$ defined in Equation~\eqref{eq:BOFA} is a level-$k$ Bayes optimal scheme for $\FFt$ for all $k\in[m]$ and any obfuscating set $H$ and obfuscating function $\mo\in\mathfrak{O}_H$; i.e., 
$\Phi^*\in\argmin_{\Phi\in\nnm}L_k(\Phi,V^*)$ for all $k\in[m]$.
\end{theorem}

For ease of reference in the sections to come, Table~\ref{tab:notation} summarizes the notation used in the paper so far.

\begin{table}[h!]
\centering
\begin{tabular}{|c|c|c|}
 \hline
 Symbol & Description & Definition\\
 \hline
 $[n]$ & For $n\in\mathbb{Z}_{>0}$, this denotes $\{1,2,3,\cdots,n\}$&-\\
 $\binom{S}{2}$
 & For a set $S$, this represents the &-\\
 & set $\{\,\{u,v\}\text{ s.t. }u,v\in S \}$&\\
$\mathcal{V}$ & Denotes a discrete set of vertex features &-\\
$\mathcal{E}$ & Denotes a discrete set of edge features &-\\
$\mre$ & The set $\mathcal{E}\cup\{\star\}$, where $\star$ is a special symbol & Def. \ref{def:richfeat}\\
& representing unavailable or missing data&\\
 $\gn$ & For $n\in\mathbb{Z}_{>0}$, the set of all labeled, undirected&-\\ &graphs on $n$ vertices&\\
 $\gnv^{\, d_1,d_2}$ & For $n,d_1,d_2\in\mathbb{Z}_{>0}$, the set of richly-featured& Def. \ref{def:richfeat} \\
 & networks of order $(n,d_1,d_2)$ with vertex (resp., edge)&\\
 & features in $\mathcal{V}^{d_1}$ (resp., $\mre^{d_2}$)&\\
 $V_g$ & For graph $g\in\gn$, $V_g$ denotes the set of vertices of $g$&-\\
 $E_g$ & For graph $g\in\gn$, $E_g$ denotes the set of edges of $g$&-\\
 $N$ & $\binom{n}{2}$&- \\
 $M$ & $\binom{m}{2}$&- \\
 $\vec{0}$ & The vector of all $0$'s&-\\
 $\bX[i,:]$ & This denotes the $i$-th row of a matrix $\bX$&- \\
 $\bX[S,:]$ & For a set $S$, this denotes the submatrix of $\bX$&-\\
 & with rows indexed by $S$& \\
 $\simeq$ & If $g_1,g_2\in\gn$ satisfy $g_1\simeq g_2$, then $g_1$ is&-\\ 
 & isomorphic to $g_2$&\\
 $\simeq$ & If $\bbg_1,\bbg_2\in\gnv^{\, d_1,d_2}$ satisfy $\bbg_1\simeq \bbg_2$, then $\bbg_1$ is&Def. \ref{def:fauto}\\ & feature-preserving isomorphic to $\bbg_2$ & \\
 $\Phi$ & A nomination scheme & Def.~\ref{def:VN} \\
 $\mathcal{N}^{(n,m)}=\mathcal{N}^{(n,m)}_{(d_1,e_1),(d_2,e_2)}$ & the set of all feature-aware VN schemes & Def.~\ref{def:VN} \\
 $L_k, L^*_k$ & level-$k$ error of a VN scheme & Def.~\ref{def:vnloss} \\
 \hline
\end{tabular}
\caption{Notation and definitions used in the remainder of the paper.}
 \label{tab:notation}
\end{table}


\section{Feature-Oblivious Vertex Nomination}
\label{sec:nofeat}

It is intuitively clear that incorporating features
should improve VN performance, provided those features
are correlated with vertex ``interestingness''.
Indeed, this is a common theme across many graph-based machine learning tasks
(see, for example, \cite{zhang2016final,binkiewicz2017covariate,meng2019co}), 
and the same holds in the present VN setting.
The combination of network structure and informative features can significantly improve the VN Bayes error.
Consider, for instance, the following simple example set in the context of the stochastic blockmodel \cite{Holland1983}. 
\begin{definition} \label{def:SBM}
An undirected $n$-vertex graph $G=(V,E)$ is an instantiation of a \emph{Stochastic Blockmodel} with parameters $(K,b,\Lambda)$, abbreviated $G\sim\text{SBM}(K,b,\Lambda)$, if:
\begin{itemize}
\item[1.] The vertex set $V$ is partitioned into $K$ communities $V=V_1\sqcup V_2\sqcup \ldots\sqcup V_K$;
\item[2.] The community membership function $b:V\rightarrow [K]$
	agrees with the partition of the vertices,
	so that $v \in V_{b(v)}$ for each $v \in V$;
\item[3.] $\Lambda$ is a $K\times K$ matrix of probabilities:
$\mathbbm{1}\left\{\{u,v\}\in E\right\}\sim\text{Bern}(\Lambda(b(u),b(v)))$
for each $\{u,v\}\in\binom{V}{2}$,
and the collection of random variables $\{\mathbbm{1}\left\{\{u,v\}\in E\right\} : \{u,v\}\in\binom{V}{2} \}$ are mutually independent.
\end{itemize}
\end{definition}

\begin{example}
\label{ex:gandf} 
Let $G_1,G_2$ be independent $2n$-vertex $\text{SBM}(2,\Lambda,b)$ random graphs with 
\begin{equation*}
\Lambda=\begin{pmatrix}
a& b\\
b& c
\end{pmatrix},\hspace{3mm}
b(v)=\begin{cases}
1&\text{ if } 1\leq v\leq n\\
2&\text{ if } n+1\leq v\leq 2n
\end{cases},
\end{equation*}
where $b<a<c$ are fixed (i.e., do not vary with $n$).
Edges in both $G_1$ and $G_2$ are independent and the probability of an edge between vertices $\{u,v\}$ is equal to
\begin{equation*}
\p( \{u,v\} \in E(G_i) )=
\begin{cases}
a&\text{ if } \{u,v\}\subset [n];\\
c&\text{ if } \{u,v\}\subset \{n+1,\cdots,2n\};\\
b&\text{ otherwise. }
\end{cases}
\end{equation*}
Take $V^*=\{v_1\}$ with corresponding vertex of interest $u^*=u_1\in V_2$ with $b(v_1)=b(u_1)=1$.
In the absence of features, $L^*_k=(1+o(1))(1-\frac{\min(k, n)}{n})$,
owing to the fact that vertices in the same community are stochastically identical.
If the graphs are endowed with one-dimensional edge features $X,Y\in\R^{2n}$, 
\begin{equation*}
X^T=Y^T=[\underbrace{1,1,\cdots,1}_{n/2\text{ total}},\underbrace{2,2,\cdots,2}_{n\text{ total}},\underbrace{1,1,\cdots,1}_{n/2\text{ total}}],
\end{equation*}
we see that a ranking scheme that ignores network structure
can do no better than randomly ranking the vertices with feature $1$,
and thus has a loss $L^*_k=1-\frac{\min(k, n)}{n}$.
In contrast, if one considers the richly attributed graphs
$(G_1,X)$ and $(G_2,Y)$, the Bayes optimal loss is improved to
$L^*_k=(1+o(1))(1-\frac{\min(k, n/2)}{n/2})$ for all $k,n$,
as the network topology and vertex features offer complementary information.
\end{example}

Can Bayes optimal performance in VN ever be improved by ignoring features?
To answer this question,
we must first define a scheme that both ignores features \emph{and} satisfies
the consistency criteria in Definition \ref{def:VN},
and it is not immediately obvious how to do so.

Firstly, we must define what it means to ignore features in a richly-featured network.
Toward this end, consider the following definition, which defines a procedure for mapping a richly-featured network to a simple graph structure (i.e., the network stripped of its feature information).
\begin{definition}
Let $N=\binom{n}{2}$.
For $\bw\in\mre^{N\times d_2}$, we define $\gamma(\bw) = (V_{\gamma(\bw)},E_{\gamma(\bw)})\in\gn$ to be the graph compatible with the edge features in $\bw$.
That is, $\gamma(\bw)\in\mathcal{G}_n$, and
$\bw$ is such that
$\bw[e,:]=(\star,\star,\cdots,\star)$ for all $e \not \in E_{\gamma(\bw)}$
and $\bw[e,:]\in\mathcal{E}^{d_2}$ for all $ e\in E_{\gamma(\bw)}.$
\end{definition}

We stress that a priori, it is not necessarily clear that a VN scheme exists that simultaneously ignores features and satisfies the consistency requirements of Definition~\ref{def:VN}.
We illustrate this point with a brief example.
A scheme $\Phi$ that ignores features must rank vertices identically regardless of features, so long as the features are compatible in the sense of the mapping $\gamma$ just defined.
More formally, it must be that for all $(g_1,g_2)\in\gn\times\gm$ and all features
\begin{equation*}
(\bx,\bw,\by,\bz),
(\bx',\bw',\by',\bz')\in\cv^{n\times d_1}\times
\widetilde\ce^{N\times e_1}
\times
\cv^{m\times d_2}
\times
\widetilde\ce^{M\times e_2}
\end{equation*}
with $\gamma(\bw)=\gamma(\bw')=g_1$ and 
$\gamma(\bz)=\gamma(\bz')=g_2$,
\begin{equation}
\label{eq:FA}
\Phi\left((g_1,\bx,\bw),\mo(g_2,\by,\bz),V^*\right)
=
\Phi\left((g_1,\bx',\bw'),\mo(g_2,\by',\bz'),V^*\right).
\end{equation}
Now consider $\bbg_1=(g_1,\bx,\bw)$ and $\bbg_2=(g_2,\by,\bz)$ that are asymmetric as attributed networks (i.e., have no non-trivial f-automorphisms),
but for which the raw graphs $g_1$ and $g_2$ have non-trivial automorphism groups.
By assumption, there exists a permutation $\sigma\neq \id$ such that 
\begin{equation*}
    \sigma(\mo(\bbg_2))=\mo(g_2,\by',\bz').
\end{equation*}
Then the consistency criterion in Definition~\ref{def:VN} requires 
\begin{equation*}
\sigma(\Phi\left(\bbg_1,\mo(\bbg_2),V^*\right))=\Phi\left(\bbg_1,\sigma(\mo(\bbg_2)),V^*\right),
\end{equation*}
while Equation~\eqref{eq:FA} requires
$   \Phi\left(\bbg_1,\sigma(\mo(\bbg_2)),V^*\right)=
\Phi\left(\bbg_1,\mo(\bbg_2),V^*\right),$
contradicting $\sigma \neq \id$.

For Equation~\eqref{eq:FA} to hold in general, we must consider a consistency criterion analogous to Equation~\eqref{eq:consis} that is compatible with schemes that ignore the vertex features.
The following definition, adapted from \cite{agterberg2019vertex}, suffices.
\begin{definition}[Feature-Oblivious VN Scheme]
\label{def:FOVN}
Let $n,m,d_1,e_1,d_2,e_2\in \Z_{>0}$ and $\cv$, $\ce$ be given.
Let $H$ be an obfuscating set of $V_1$ and $V_2$ of order $|V_2|=m$, and let $\mo\in\mathfrak{O}_H$ be given.
A \emph{feature-oblivious vertex nomination scheme} (FO-VN scheme) is a function 
\begin{equation*}
\Psi: \gnv^{\,d_1,e_1} \times \mo(\gmv^{\,d_2,e_2}) \times 2^{V_1} \rightarrow \calT_{H},
\end{equation*}
satisfying Equation~\eqref{eq:FA} as well as
the consistency criterion in Assumption~\ref{ass:FOVN} below.
\end{definition}

Similar to the FA-VN consistency criteria, we require FO-VN schemes to be similarly label-agnostic for the obfuscated labels of the second graph.
\begin{assu}[FO-VN Consistency Criteria]
\label{ass:FOVN}
With notation as in Definition~\ref{def:FOVN},
for each $\bbg_2=(g_2,\by,\bz)\in\gmv^{d_2,e_2}$ and $u\in V_{2}$, let
\begin{equation*}
\mathcal{J}(u;\bbg_2)=\{w\in V_2\text{ s.t. } 
	\exists\text{ an automorphism }\sigma\text{ of }g_2,\text{ s.t. }\,\sigma(u)=w\}.
\end{equation*}
For any ${\bf g_1}\in\gnv^{\,d_1,e_1},$ ${\bf g_2}\in\gmv^{\,d_2,e_2}$, $V^*\subset V_1\cap V_2$, obfuscating functions $\mo_1,\mo_2\in\mathfrak{O}_H$ and any $u\in V_2$, we require
\begin{equation} \label{eq:consis2}
\begin{aligned}
&\mathfrak{r}_{\Psi}({\bf g_1},{\bf g_2},\mo_1,V^*,\mathcal{J}(u;{\bf g_2}))=\mathfrak{r}_{\Psi}({\bf g_1},{\bf g_2},\mo_2,V^*,\mathcal{J}(u;{\bf g_2})) \\
 &~~~\Leftrightarrow \\
&\forall k\in[m]:  \\
&~~~\mo_2\circ\mo_1^{-1}\big( \mathcal{J}(\Psi({\bf g_1},\mo_1({\bf g_2}),V^*)[k]);\mo_1({\bf g_2})\big)=\mathcal{J}\left( \Psi({\bf g_1},\mo_2({\bf g_2}),V^*)[k];\mo_2({\bf g_2}) \right) 
\end{aligned}
\end{equation}
where $\Psi({\bf g_1},\mo({\bf g_2}),V^*)[k]$ denotes the $k$-th element
in the ordering $\Psi({\bf g_1},\mo({\bf g_2}),V^*)$
(i.e., the rank-$k$ vertex).
\end{assu}

The criterion in Equation~\eqref{eq:consis2} is less restrictive than that in Equation~\eqref{eq:consis}, 
and it is not immediate that incorporating features yields an FA-VN scheme with smaller loss than the Bayes optimal FO-VN scheme.
We illustrate this in the following example.
\begin{example}
Let $F\in\mathcal{F}^{(n,m)}$ be a distribution such that $G_1\stackrel{a.s.}{=}K_n$
and $G_2\stackrel{a.s.}{=}K_m$, where $K_n$ denotes the complete graph on $n$ vertices.
If the $f$-automorphism group of ${\bf G}_1$ ${\bf G}_2$ are a.s.\ trivial, then any given FA-VN scheme can be outperformed by a well-chosen FO-VN scheme.
Indeed, if there is a single vertex of interest $v^*$ in ${\bf G}_1$ with corresponding vertex $u^*$ in ${\bf G}_2$, then there exists a FO-VN scheme $\Psi$ that satisfies $\Psi(\bbg_1,\mo(\bbg_2),v^*)[1]=\mo(u^*)$ for almost all $\bbg_1,\mo(\bbg_2)$.
Such a $\Psi$ cannot satisfy Equation~\eqref{eq:consis}, and it is possible to have $L_1(\Phi^*)>0$ for FA-VN Bayes optimal $\Phi^*$.
\end{example}

\noindent However, consider distributions satisfying the following assumption.
\begin{assu} \label{ass:A2}
Let $({\bf G}_1=(G_1,\bX,\bW),{\bf G}_2=(G_2,\bY,\bZ))\sim F^{(n,m)}_{\Theta}\in\mathcal{F}^{(n,m)}$
and consider the events
\begin{equation*} \begin{aligned}
D_3&=\{\text{ the only automorphism of } G_1\text{ is }\sigma=\id_n  \}\\
D_4&=\{\text{ the only automorphism of }G_2\text{ is }\sigma=\id_m  \}.
\end{aligned} \end{equation*} 
$F^{(n,m)}_{\Theta}$
satisfies
$\p_{\FFt}( D_3)=\p_{\FFt}( D_4)=1$.
\end{assu}
Under Assumption \ref{ass:A2},
we have that $\mathcal{I}(u;\bbg_2)\stackrel{a.s.}{=}\mathcal{J}(u;\bbg_2)\stackrel{a.s.}{=}\{u\}$,
and the consistency criteria in Assumptions \ref{ass:FAVN} and \ref{ass:FOVN} are almost surely equivalent.
It is then immediate that Bayes optimality cannot be improved by ignoring features. That is, a FO-VN $\Psi$ scheme is almost surely a FA-VN scheme, and hence $L_k(\Psi,V^*)\geq L_k(\Phi^*,V^*)$ for all $k\in[m]$.
This leads us to ask whether we can establish conditions under which ignoring features \emph{strictly} decreases VN performance.

\subsection{Feature oblivious Bayes optimality}
We first establish the notion of a Bayes optimal FO-VN scheme for distributions satisfying Assumption \ref{ass:A2}.
Defining
\begin{align*}
\mathfrak{G}_{n,m} 
&=\{(g_1,\mo(g_2))\in \gn \times \mo(\gm)\text{ s.t. }g_1,g_2\text{ are asymmetric} \},
\end{align*}
let $\{(g_1,g_2)\}_{i=1}^p$ be such that
$\{(g_1,[\mo(g_2)])\}_{i=1}^p$
partitions
$\mathfrak{G}_{n,m}$.
For $\FFt$ supported on 
$\mathfrak{G}_{n,m}$,
it follows from \cite{agterberg2019vertex} that a Bayes optimal FO-VN scheme, $\Psi^*$,
can be constructed as follows.
If $(g_1^i,g_2^i)\in\{(g_1,g_2)\}_{i=1}^p$, and $(\bbg_1^i , \mo(\bbg_2^i))$ is {\em any} featured extension of $(g_1^i,g_2^i)$ (note that this notation will be implicit below), we sequentially define (breaking ties in a fixed but arbitrary manner and writing $\p(\cdot)$ for $\p_{\FFt}(\cdot)$ to ease notation)
\begin{equation*} 
\begin{aligned} 
\Psi^*(\bbg_1^i , \mo(\bbg_2^i), V^*)[1] 
&\in \argmax_{\substack{u\in H }}   \,\,
\p\bigg( (g_1^i, [\mo(g_2^i)])_{u\in \mo(V^*) }\,\, \big|\,\, (g_1^i, [\mo(g_2^i)]) \bigg) \\
\Psi^*(\bbg_1^i , \mo(\bbg_2^i), V^*)[2] 
&\in \argmax_{\substack{u\in H\setminus\{ \Psi^*[1]\}}}  
\p\bigg( (g_1^i, [\mo(g_2^i)])_{u\in \mo(V^*) }\,\, \big|\,\, (g_1^i, [\mo(g_2^i)]) \bigg) \\
&\vdots \\
\Psi^*(\bbg_1^i , \mo(\bbg_2^i), V^*)[m] 
&\in \argmax_{\substack{u\in H\setminus\{\cup_{j< m}\{\Psi^*[j]\}}}
\p\bigg( (g_1^i, [\mo(g_2^i)])_{u\in \mo(V^*) }\,\, \big|\,\, (g_1^i, [\mo(g_2^i)]) \bigg) ,
\end{aligned} \end{equation*}
where $(g_1, [\mo(g_2)])$ and $(g_1, [\mo(g_2)])_{u\in \mo(V^*)}$ (that is, the graphs without
their features) are defined analogously to the featured $(\bbg_1, [\mo(\bbg_2)])$ and 
$(\bbg_1, [\mo(\bbg_2)])_{u\in \mo(V^*)}$ respectively:
\begin{align*}
({ g_1}, [\mo(g_2)]) &= \bigg\{ (g_1, \widehat g_2) \in \mathfrak{G}_{n,m}\text{ s.t. } \mo(\widehat g_2) \simeq \mo(g_2) \bigg\} \\
({ g_1}, [\mo(g_2)])_{u = \mo(v)} 
&=\bigg\{ (g_1, \widehat g_2) \in \mathfrak{G}_{n,m}\text{ s.t. } \mo(\widehat g_2) = \sigma( \mo(g_2)), \text{ where $\sigma$ is an}\\
& \hspace{15mm}\text{ isomorphism  satisfying }\sigma(u)=\mo(v) \bigg\}. 
\end{align*}
For each
$(\bbg_1',\bbg_2')\in(\bbg_1^i, [\mo(\bbg_2^i)]),$ choose the f-isomorphism $\sigma$ such that $\mo(\bbg_2')=\sigma(\mo(\bbg_2^i))$, and define
$${\Psi^*}(\bbg_1',\mo(\bbg_2'),V^*)=\sigma({\Psi^*}(\bbg_1^i,\mo(\bbg_2^i),V^*)).$$
For elements $(\bbg_1, \mo(\bbg_2))\notin\mathfrak{G}_{n,m}$, any fixed and arbitrary definition of $\Phi_O^*$ satisfying Equation~\eqref{eq:consis2} suffices (as this set has measure $0$ under $\FFt$ by Assumption \ref{ass:A2}).
Note that $\Psi^*$ is almost surely well-defined, as the definition of $\Psi^*$ on 
\begin{equation*} \begin{aligned}
\label{eq:assym}
\mathfrak{A}_{n,m}=\AS
:=\{&(\bbg_1=(g_1,\bx,\bw),\bbg_2=(g_2,\by,\bz))\in \gnv^{\,d_1,e_1} \times \gmv^{\,d_2,e_2} \\
&~~~~~~\text{ s.t. }g_1,g_2\text{ are asymmetric} \}
\end{aligned} \end{equation*}
 is independent of the choice of the partition $\{(g_1^i,g_2^i) : i =1,2,\dots,p \}$.

\subsection{The Benefit of Features}
\label{sec:features}

With the FO-VN scheme defined, we seek to understand when, for distributions $F^{(n,m)}_\Theta$ supported on $\mathfrak{A}_{n,m}$, we have $L_k(\Phi^*,V^*)<L_k(\Psi^*,V^*)$ where $\Phi^*$ and $\Psi^*$
are the Bayes optimal FA-VN and FO-VN schemes, respectively, under $F^{(n,m)}_\Theta$.
Toward this end, we first define the following $\mathcal{T}_H$-valued random variable.
\begin{definition} \label{def:XPhi}
Let $F^{(n,m)}_\Theta\in\mathcal{F}^{(n,m)}$, let $V^* \subseteq V_1\cap V_2$ be a given set of vertices of interest, and let $H$ be an obfuscating set of $V_1$ and $V_2$ of order $|V_2|=m$ with obfuscating function $\mo\in\mathfrak{O}_H$.
Let $\Phi$ be a VN scheme (either feature-aware or feature-oblivious), and define, letting $\Omega$ be our sample space, the $\mathcal{T}_H$-valued random variable 
\begin{equation*}
    X_{\Phi}:\Omega\mapsto \mathcal{T}_H
\end{equation*}
by
$X_\Phi(\omega)=\Phi({\bf G}_1(\omega),\mo({\bf G}_2(\omega)),V^*)$.
For each $k\leq m$, define 
$X_{\Phi}^k=X_\Phi[1:k]\in\mathcal{T}_H^k$, where we define $\mathcal{T}_H^k$ to be the set of all $k$-tuples of distinct elements of $H$ (each such tuple can be viewed as specifying a total ordering of $k$ distinct elements of $H$).
\end{definition}
\begin{remark} \label{rem:measurability}
Note that in the setting of continuous features, the measurability of $X_\Phi$ is not immediate (and indeed, is non-trivial to establish); this technical hurdle is the main impetus for discretizing the feature space.
\end{remark}
We can now characterize the conditions under which incorporating features strictly improves VN performance.  A proof of Theorem~\ref{thm:feat} can be found in Appendix~\ref{sec:featpf}.
\begin{theorem} \label{thm:feat}
Consider the setup and notation of Definition~\ref{def:XPhi}
and suppose that $F^{(n,m)}_\Theta\in\mathcal{F}^{(n,m)}$ satisfies Assumption~\ref{ass:A2}.
Letting $\Phi^*$ and $\Psi^*$ be Bayes optimal FA-VN and FO-VN schemes, respectively,
under $F^{(n,m)}_\Theta$, we have that
$L_k(\Phi^*,V^*)=L_k(\Psi^*,V^*)$ if and only if there exists a Bayes optimal FA-VN scheme $\Phi^*$ with 
$$\mathbb{I}(X_{\Phi^*}^k; (G_1,G_2) )=\mathbb{H}(X_{\Phi^*}^k),$$
where $\mathbb{I}$ is the mutual information and $\mathbb{H}$ the statistical entropy defined by
\begin{align*}
&\mathbb{H}(X_{\Phi^*}^k)=-\sum_{\xi\in \mathcal{T}_H^k}\p(X_{\Phi^*}^k=\xi)\log(\p(X_{\Phi^*}^k=\xi))\\
&\mathbb{I}(X_{\Phi^*}^k; (G_1,G_2) )=
\sum_{\xi\in \mathcal{T}_H^k}\sum_{\substack{(g_1,g_2)\\ \in\gn\times\gm}}\p(\xi, (g_1,g_2))\log\left(\frac{\p(\xi, (g_1,g_2))}{\p(\xi)\p((g_1,g_2))}\right),
\end{align*}
where we have written $\p(\xi, (g_1,g_2))$ as shorthand for
$\p(X_{\Phi^*}^k=\xi, (G_1,G_2)=(g_1,g_2))$.
\end{theorem}

\noindent We note that, since $\mathbb{I}(X_{\Phi^*}^k; (G_1,G_2) )\leq \mathbb{H}(X_{\Phi^*}^k)$,
we can restate the result of
Theorem \ref{thm:feat} as
$L_k(\Phi^*,V^*)<L_k(\Psi^*,V^*)$ if and only if for all Bayes optimal FA-VN schemes $\Phi^*$,  
$$\mathbb{I}(X_{\Phi^*}^k; (G_1,G_2) )<\mathbb{H}(X_{\Phi^*}^k).$$
Stated succinctly, there is excess uncertainty in $X_{\Phi^*}^k$ after observing $(G_1,G_2)$; and $X_{\Phi^*}^k$ is not deterministic given $(G_1,G_2)$.
 

\section{Network-Oblivious Vertex Nomination}
\label{sec:nogs}
In contrast to the feature-oblivious VN schemes considered in Section \ref{sec:nofeat}, one can also consider VN schemes that use only features and ignore network structure.
Defining such a network-oblivious VN scheme (NO-VN scheme) is not immediately straightforward.
Ideally, we would like to have that
for all $(g_1,g_2),(g_1',g_2')\in\gn\times\gm$ and all
edge features $(\bw,\bz)$, $(\bw',\bz')$ compatible with $(g_1,g_2)$ and $(g_1',g_2')$ respectively,
\begin{equation}
\label{eq:NOVN}
\Phi\left((g_1,\bx,\bw),\mo(g_2,\by,\bz),V^*\right)
=
\Phi\left((g_1',\bx,\bw'),\mo(g_2',\by,\bz'),V^*\right)
\end{equation}
for any choice of vertex features $\bx,\by$.
As in the FO-VN scheme setting, this leads to potential violation of the internal consistency criteria of Equation~\eqref{eq:consis}.
Indeed, consider $\bbg_1=(g_1,\bx,\bw)$ and $\bbg_2=(g_2,\by,\bz)$ with asymmetric graphics
but with symmetries in $\by$ (i.e., there exists non-identity permutation matrix $P_\sigma$ such that $P_\sigma\by=\by$). On such networks, Equations~\eqref{eq:consis} and~\eqref{eq:NOVN} cannot both hold simultaneously.
Thus, we consider a relaxed consistency criterion as in Assumption \ref{ass:FOVN}. We first define
\begin{equation*}
\mathcal{Y}(u;\bbg_2)=\{w\in V(g_2) : 
	\exists\text{ bijection }\sigma
	\text{  s.t. }P_{\sigma}\by=\by\text{ and }\,\sigma(u)=w\},
\end{equation*}
and make the following consistency assumption.
\begin{assu}[NO-VN Consistency Criteria]
\label{ass:NOVN}
For any ${\bf g_1}\in\gnv^{\,d_1,e_1},$ ${\bf g_2}\in\gmv^{\,d_2,e_2}$,
letting $H$ be an obfuscating set of $V_1$ and $V_2$ of order $|V_2|=m$
with $\mo_1,\mo_2\in\mathfrak{O}_H$,
$V^*\subset V_1\cap V_2$ be the set of vertices of interest, and taking $u \in V(g_2)$, if $\Xi$ is a VN scheme satisfying this assumption, then 
\begin{align}
\label{eq:consis3}
&\mathfrak{r}_{\Xi}({\bf g_1},{\bf g_2},\mo_1,V^*,\mathcal{Y}(u;{\bf g_2}))=\mathfrak{r}_{\Xi}({\bf g_1},{\bf g_2},\mo_2,V^*,\mathcal{Y}(u;{\bf g_2}))\\
\notag &\Leftrightarrow \mo_2\circ\mo_1^{-1}\big( \mathcal{Y}(\Xi({\bf g_1},\mo_1({\bf g_2}),V^*)[k]);\mo_1({\bf g_2})\big)=\mathcal{Y}\left( \Xi({\bf g_1},\mo_2({\bf g_2}),V^*)[k];\mo_2({\bf g_2}) \right)\\
& \hspace{25mm}\text{ for all }k\in[m],\notag
\end{align}
where $\Xi({\bf g_1},\mo({\bf g_2}),V^*)[k]$ denotes the $k$-th element in the ordering $\Xi({\bf g_1},\mo({\bf g_2}),V^*)$ (i.e., the rank-$k$ vertex under $\Phi$).
\end{assu}
A network-oblivious VN scheme $\Xi$ is then a VN scheme as in Definition~\ref{def:VN}, where the consistency criterion of Equation~\eqref{eq:consis} is replaced with that in Equation~\eqref{eq:consis3} and we further require Equation~\eqref{eq:NOVN} to hold.
As with FA-VN schemes, we consider distributions satisfying the following assumption.
\begin{assu}
\label{ass:A3}
Let $( (G_1,\bX,\bW),(G_2,\bY,\bZ))\sim F^{(n,m)}_{\Theta}\in\mathcal{F}^{(n,m)}$
and define the events
$D_5=\{ \bX=P_\sigma\bX \implies \sigma = \id \}$
and
$D_6=\{ \bY=P_\sigma\bY \implies \sigma = \id \}$.
$F^{(n,m)}_{\Theta}$ is such that
$\p_{\FFt}( D_5)=\p_{\FFt}( D_6)=1$.
\end{assu}
We note the parallel between this assumption and Assumption~\ref{ass:A2}, while noting that the two assumptions concern permutations acting on markedly different objects (graphs, in the case of Assumption~\ref{ass:A2} and vertex-level features in the case of Assumption~\ref{ass:A3}).
Under Assumption~\ref{ass:A3}, we have that $\mathcal{I}(u;\bbg_2)\stackrel{a.s.}{=}\mathcal{Y}(u;\bbg_2)\stackrel{a.s.}{=}\{u\}$,
and the consistency criteria of Assumption~\ref{ass:NOVN} are almost surely equivalent.
As in Section \ref{sec:nofeat}, under this assumption, we have that Bayes optimality cannot be improved by ignoring the network. Indeed, one can show that the NO-VN scheme is almost surely a FA-VN scheme, and we are led once again to ask under what circumstances VN performance will be strictly worsened by ignoring the network (and subsequently, the edge features).
To this end, we wish to compare Bayes optimality of NO-VN with that of FA-VN.

\subsection{Network-oblivious Bayes optimality}
\label{sec:NOBO}
We first establish the notion of a Bayes optimal NO-VN scheme for distributions satisfying Assumption \ref{ass:A3}. Define
\begin{align*}
\mathfrak{F}_{n,m}
:=\{(\bx,\by)\in \mathbb{R}^{n\times d_1} \times \mathbb{R}^{m\times d_2}\text{ s.t. }\bx,\by\text{ have distinct rows} \},
\end{align*}
and for $(\bx,\by)\in\mathfrak{F}_{n,m}$, define
\begin{align*}
(\bx, [\mo(\by)]) &= \bigg\{ (\bx, \widehat \by) \in \mathfrak{F}_{n,m}\text{ s.t. there exists permutation }\sigma\text{ s.t. }P_\sigma \by=\widehat\by \bigg\} \\
(\bx, [\mo(\by)])_{u = \mo(v)} 
&=\bigg\{ (\bx, \widehat \by) \in \mathfrak{F}_{n,m}\text{ s.t. there exists permutation }\sigma\text{ s.t. }P_\sigma \by=\widehat\by\\
&\hspace{20mm}\text{ and $\sigma$ satisfies }\sigma(u)=\mo(v)\bigg\}.
\end{align*}
For a given $\FFt$ satisfying Assumption \ref{ass:A3}, we will define the Bayes optimal NO-VN scheme, $\Xi^*$, element-wise on vertex feature matrices with no row repetitions (similar to in Section \ref{sec:BO}),
and then lift the scheme to all richly-featured graphs with vertex features not in $\mathfrak{F}_{n,m}$.
For $(\bx,\bw,\by,\bz)$ with $\bx$ and $\by$ having distinct rows, let $g_1$ and $g_2$ be the unique graphs with edge structure compatible with $\bw$ and $\bz$ respectively.
Writing $\bbg_1=(g_1,\bx,\bw)$ and $\bbg_2=(g_2,\by,\bz)$, we define,
writing $\p(\cdot)$ for $\p_{\FFt}(\cdot)$ to ease notation 
\begin{equation*} 
\begin{aligned} 
\Xi^*(\bbg_1 , \mo(\bbg_2), V^*)[1] 
&\in \argmax_{\substack{u\in H }}   \,\,
\p\bigg( (\bx, [\mo(\by)])_{u\in \mo(V^*) }\,\, \big|\,\, (\bx, [\mo(\by)]) \bigg) \\
\Xi^*(\bbg_1 , \mo(\bbg_2), V^*)[2] 
&\in \argmax_{\substack{u\in H\setminus\{ \Xi^*[1]\}}}  
\p\bigg( (\bx, [\mo(\by)])_{u\in \mo(V^*) }\,\, \big|\,\, (\bx, [\mo(\by)]) \bigg) \\
&\vdots \\
\Xi^*(\bbg_1 , \mo(\bbg_2), V^*)[m] 
&\in \argmax_{\substack{u\in H\setminus\{\cup_{j< m}\{\Xi^*[j]\}}}
\p\bigg( (\bx, [\mo(\by)])_{u\in \mo(V^*) }\,\, \big|\,\, (\bx, [\mo(\by)]) \bigg) ,
\end{aligned} \end{equation*}
where we write 
$(\bx, [\mo(\by)])$ in the conditioning statement as shorthand for
(writing $\bbg_1=(g_1,\bx,\bw)$ and $\bbg_2=(g_2,\by,\bz)$)
\begin{equation*}
    \left({\bf G}_1,\mo({\bf G}_2) \right) \in\left\{((g_1,\bx,\bw),\mo((g_2,\by,\bz)))\text{ s.t. }(\bx,\by)\in(\bx, [\mo(\by)])\right\}.
\end{equation*}
Note that once again, ties in the maximizations when constructing $\Xi^*$ are assumed to be broken in an arbitrary but nonrandom manner.
For each element 
\begin{equation*}
    (\bbg_1',\bbg_2')\in(\bbg_1, [\mo(\bbg_2)]),
\end{equation*}
choose the f-isomorphism $\sigma$ such that $\mo(\bbg_2')=\sigma(\mo(\bbg_2))$, and define
\begin{equation*}
{\Xi^*}(\bbg_1',\mo(\bbg_2'),V^*)=\sigma({\Xi^*}(\bbg_1,\mo(\bbg_2),V^*)).
\end{equation*}
For elements $(\bx,\by)\notin\mathfrak{F}_{n,m}$ and arbitrary edge features $\bw,\bz$, any fixed and arbitrary definition of $\Xi^*$ on (well-defined) graphs in $\gn\times\{\bx\}\times\{\bw\}\times\gm\times\{\by\}\times\{\bz\}$ suffices, subject to the internal consistency criterion in Equation~\eqref{eq:consis3}, as this set has measure $0$ under $\FFt$ under Assumption \ref{ass:A3}.

\subsection{The benefit of network topology}
\label{sec:ben_net}
Once again, for distributions satisfying Assumption \ref{ass:A3},
our aim is to under stand when $L_k(\Phi^*,V^*)<L_k(\Xi^*,V^*)$.
That is, when does incorporating the network topology into the vertex-level features strictly improve VN performance? Theorem~\ref{thm:topology} characterizes these conditions.
The proof is completely analogous to the proof of Theorem~\ref{sec:vnfeat}, and is included in Appendix \ref{sec:top_proof} for completeness.

\begin{theorem}
\label{thm:topology}
Let $F^{(n,m)}_\Theta\in\mathcal{F}^{(n,m)}$ be a richly-featured nominatable distribution satisfying Assumption \ref{ass:A3}.
Let $V^* \subseteq V_1\cap V_2$ be a given set of vertices of interest and let $H$ be an obfuscating set of $V_1$ and $V_2$ of order $|V_2|=m$ with $\mo\in\mathfrak{O}_H$.
Let $\Phi^*$ and $\Xi^*$ be Bayes optimal FA-VN and NO-VN schemes, respectively, under $F^{(n,m)}_\Theta$. Then
$L_k(\Phi^*,V^*)=L_k(\Xi^{*},V^*)$ if and only if there exists a Bayes optimal FA-VN scheme $\Phi^*$ with 
\begin{equation*}
    \mathbb{I}(X_{\Phi^*}^k; (\bX,\bY) )=\mathbb{H}(X_{\Phi^*}^k).
\end{equation*}
\end{theorem}

\noindent Note that, since $\mathbb{I}(X_{\Phi^*}^k; (\bX,\bY) )\leq \mathbb{H}(X_{\Phi^*}^k),$ 
Theorem~\ref{thm:topology} can be restated as
$L_k(\Phi^*,V^*)<L_k(\Xi^*,V^*)$ (i.e., incorporating network structure improves performance) if and only if all Bayes optimal FA-VN schemes $\Phi^*$ satisfy
$\mathbb{I}(X_{\Phi^*}^k; (\bX,\bY) )<\mathbb{H}(X_{\Phi^*}^k)$.
Said yet another way, incorporating network structure improves VN performance if and only if there is excess uncertainty in $X_{\Phi^*}^k$ conditional on the features $(\bX,\bY)$.
This is precisely when the network structure is informative--- the FA-VN scheme $\Phi^*$ incorporates both network and feature information into its ranking, while the NO-VN scheme incorporates only the feature information carried by $(\bX,\bY)$.

\section{Simulations and Experiments}
\label{sec:exp}

We turn now to a brief experimental exploration of the VN problem as applied to both simulated and real data.
We consider a VN scheme based on spectral clustering, which we denote $\vnase$.
We refer the reader to \cite{agterberg2019vertex} for details and a further exploration of this scheme in an adversarial version of vertex nomination without node or edge features. 

In our experiments, edge features will appear as edge weights or edge directions, while vertex features will take the form of feature matrices $\bx$ and $\by$,
following the notation of previous sections.
The scheme $\vnase$ proceeds as follows.
Note that we have assumed $n=m$ for simplicity, but the procedure can be extended to pairs of differently-sized networks in a straight-forward manner.
\begin{itemize}
    \item[i.] Pass the edge weights to ranks, and augment the diagonal of the adjacency matrix by setting $A_{i,i} =\sum_{j\neq i} A_{i,j}/(n-1)$ \cite{tang2018connectome}; see Appendix \ref{sec:ptrda} for detail.
    \item[ii.] Embed the two networks into a common Euclidean space, $\mathbb{R}^d$
        using Adjacency Spectral Embedding \cite{sussman2014consistent}; see Appendix \ref{sec:ASE} for details.

    The embedding dimension $d$ is chosen by estimating the elbow in the scree plots of the adjacency matrices of the networks $G_1$ and $G_2$ \cite{zhu2006automatic}, taking $d$ to be the larger of the two elbows.

    Applying ASE to an $n$-vertex graph results in a mapping of the $n$ vertices in the graph
    to points in $\mathbb{R}^d$. We denote the embeddings of graphs $G_1$ and $G_2$ by $\widehat{X}_1,\widehat{X}_2 \in \mathbb{R}^{n \times d}$, respectively, with the $i$-th row of each of these matrices corresponding to the embedding of the $i$-th vertex in its corresponding network.
    
    \item[iii.] Given seed vertices $S$ (see Appendix \ref{sec:seeds}) whose correspondence is known a priori across networks, solve the orthogonal Procrustes problem \cite{GowDij2004} (see Appendix \ref{sec:proc})
to align the rows of $\widehat{X}_1[S,:]$ and $\widehat{X}_2[S,:]$. Apply this Procrustes rotation to the rows of $\widehat{X}_2$, yielding $\widehat{Y}_2 \in \mathbb{R}^{n \times d}$.
    If $p$-dimensional vertex features are available, append the vertex features to the embeddings as $Z_1=[\widehat X_1\,|\,\bx] \in \Rb^{n \times (d+p)}$ and
    $Z_2=[\widehat Y_2\,|\,\by] \in \Rb^{n \times (d+p)}$.
    
    \item[iv.] Cluster the rows of both $Z_1$ and $Z_2$ using a Gaussian mixture modeling-based clustering procedure; see, for example, the \texttt{mClust} package in \texttt{R}; \cite{mclust}.

    For each vertex $v$, let $\mu_v$ and $\Sigma_v$ be the mean and covariance of the normal mixture component containing $v$.
    For each $u\in V(G_2)$, compute the distances 
    \begin{equation*}
        D(V^*,u)=\min_{v^*\in V^*}\max\left\{
    \sqrt{(v^*-u)\Sigma_u^{-1}(v^*-u)},
    \sqrt{(v^*-u)\Sigma_{v^*}^{-1}(v^*-u)}
    \right\}.
     \end{equation*}
    
    \item[v.] Rank the unseeded vertices in $G_2$ so that the vertex $u$ minimizing $D(V^*,u)$ is ranked first, with ties broken in an arbitrary but fixed manner.
\end{itemize}
Below, we apply this VN scheme in an illustrative simulation and in two real data network settings derived from neuroscience and text-mining applications.

\subsection{Synthetic data}
\label{sec:sim}
To further explore the complementary roles of network structure and features in vertex nomination, we consider the following simulation, set in the context of the stochastic blockmodel \cite{Holland1983}, as described in Definition~\ref{def:SBM}.
We consider 
$G_1\sim \text{SBM}(5,b,\Lambda_1)$
independent of
$G_2\sim \text{SBM}(5,b,\Lambda_2),$
with $V(G_i)=\{1,2,\ldots,250\}$,
$b(v)=\lceil 250/v\rceil$, 
\begin{equation*}
    \Lambda_1=
\operatorname{diag}(\epsilon+0.05,\epsilon,\epsilon,\epsilon,\epsilon)+0.3*J_5,
\end{equation*}
and $\Lambda_2=0.8*\Lambda_1+0.2*J_5$, where $J_p$ denotes the $p$-by-$p$ matrix of all ones.
We designate block 1 as the {\em anomalous block},
containing the vertices of interest across the two networks,
with the signal in the anomalous block 1 dampened in $G_2$ compared to $G_1$
owing to the convex combination of $\Lambda_1$ and the ``flat'' matrix $J_5$.
We will consider the vertices of interest to be all $v\in V$ such that $b(v)=1$.
We select 10 vertices at random from block 1 in $G_1$ and from block 1 in $G_2$ to serve as ``seeded'' vertices, meaning vertices whose correspondences are known ahead of time.

We consider vertex features $\bx,\by\in\mathbb{R}^{250\times 5}$ of the form (letting $I_d$ denote the $d$-by-$d$ identity matrix)
\begin{equation*}
    \bx(v) \sim
\begin{cases}
\operatorname{Normal} (\delta\vec 1,I_5)\text{ if } b(v)=1\\
\operatorname{Normal} (\vec 0,I_5)\text{ if } b(v)\neq 1
\end{cases}
\hspace{2mm}
\by(v) \sim
\begin{cases}
\operatorname{Normal} (\delta\vec 1,I_5)\text{ if } b(v)=1\\
\operatorname{Normal} (\vec 0,I_5)\text{ if } b(v)\neq 1
\end{cases}
\end{equation*}
independently over all $v \in V$ and generating $\bx$ and $\by$ independently of one another.
Note that when applying our $\VNA$ scheme to the above data, we set the number of blocks to be the ``true'' $d=5$, with the number of clusters in step (iv) set to 5 as well.
In practice, there are numerous principled heuristics to select this dimension parameter (e.g., USVT or finding an elbow in the scree plot \cite{chatterjee2014matrix,zhu2006automatic})
and the number of clusters (e.g., optimizing silhouette width or minimizing BIC \cite{mclust}).
We do not pursue these model selection problems further here.

The effects of $\epsilon$ and $\delta$ are as follows.
Larger values of $\epsilon$ provide more separation between the blocks in the underlying SBM, making it easier to distinguish the vertices of interest from the seeded vertices.
This is demonstrated in Figure~\ref{fig:varyeps}, where we vary $\epsilon = 0,0.1,0.2,0.3,0.5$ with $\delta=1$ held fixed.
The figure shows, for different values of $k$, the gain in precision at $k$ achieved by incorporating the graph topology as compared to a nomination scheme based on features alone.
That is, defining
\begin{itemize}
\item $\rGF(k)$ to be the number of vertices of interest in $G_2$ nominated in the top $k$ by $\VNA$ applied to $(G_1,\bx,G_2,\by)$, and
\item $\rF(k)$ to be the number of vertices of interest in $G_2$ nominated in the top $k$ by $\VNA$ applied to $(\bx,\by)$, that is, step (iv) of the algorithm above applied only to the vertex features,
\end{itemize}
That is, letting $\Phi$ be any of the our nomination schemes under consideration (e.g., the features-only scheme) and letting $V^*(G_2)$ denote the vertices of interest in $G_2$, we evaluate performance according to the number of vertices ranked in the top 
\begin{equation*}
r(k)= \left| \{u \in V^*(G_2)\in\text{ s.t. } \rank_{\Phi}(v)\leq k\} \right|,~~~k=1,2,\dots,40.
\end{equation*}
We note that we do not consider seeded vertices in our ranked list,
so the maximum value achievable by either $\rGF$ or $\rF$ is 40.

Figure~\ref{fig:varyeps} plots $\rGF(k)-\rF(k)$ for $k\in(1,10,20,30,40)$.
Results are averaged over 100 Monte Carlo replicates of the experiment,
with error bars indicating two standard errors of the mean.
Examining the figure, we see the expected phenomenon:
as  $\epsilon$ increases, the gain in VN precision from incorporating the network increases.
For small values of $\epsilon$, the graphs are detrimental to performance when compared to using features alone, since the structure of $\Lambda_1$ and $\Lambda_2$ are such that it is difficulty to distinguish the communities from one another (and to distinguish the interesting community from the rest of the network).
As $\epsilon$ increases, the community structure in networks $G_1$ and $G_2$ becomes easier to detect, and incorporating network structure into the VN procedure becomes beneficial to performance as compared to a procedure using only vertex features.

\begin{figure}[t!]
    \centering
    \includegraphics[width=0.7\textwidth]{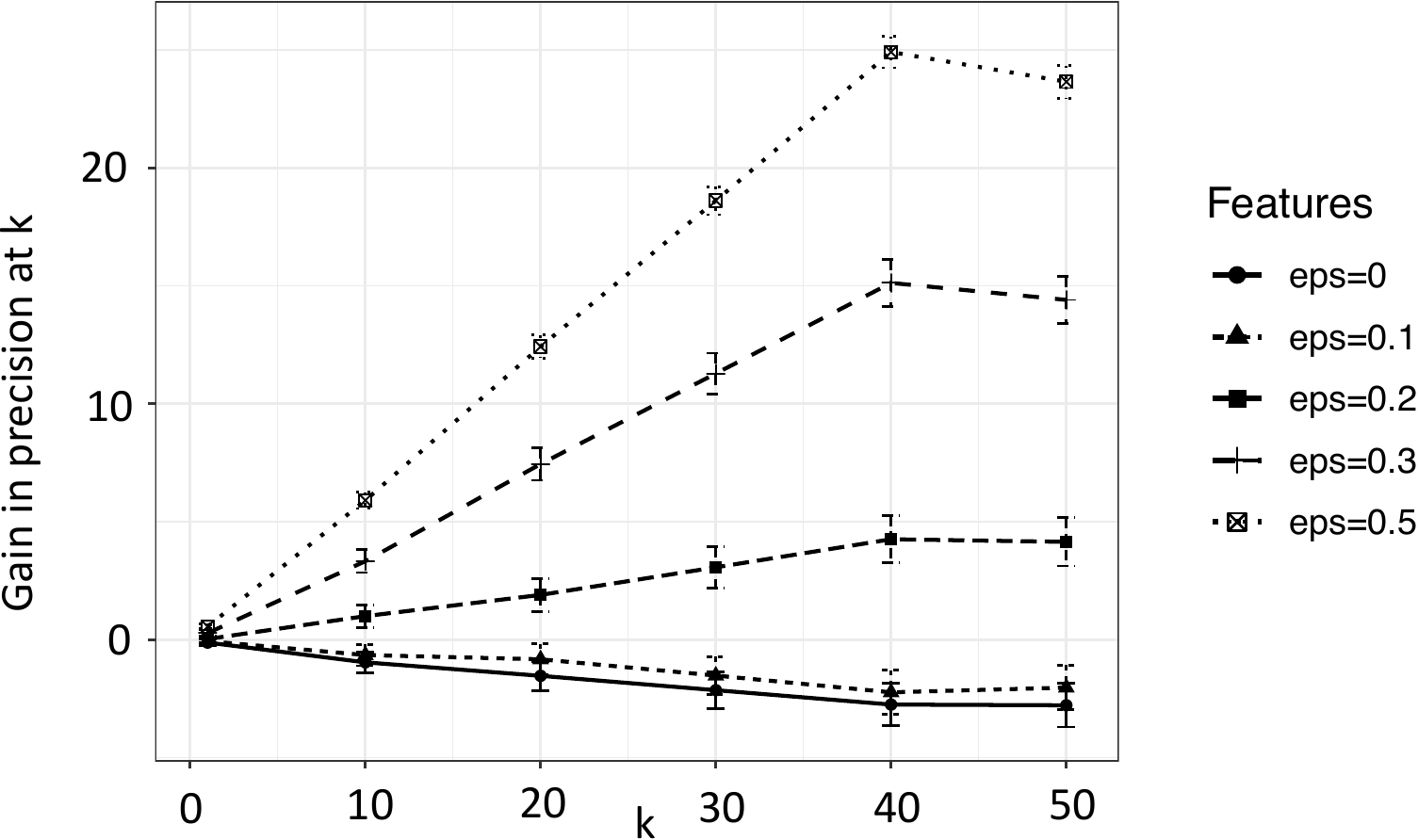} 
        \caption{ Improvement in vertex nomination performance
        under the stochastic block model specified above,
        as a function of $\epsilon=0,0.1,0.2,0.3,0.5$ for fixed $\delta=1$,
        based on $10$ randomly chosen seeded vertices.
        The plot shows $\rGF(k)-\rF(k)$, as defined in Equation~\eqref{eq:consis},
        for $k\in \{ 1,10,20,30,40 \}$.
    Results are averaged over 100 Monte Carlo trials, with error bars indicating two standard errors of the mean. }
    \label{fig:varyeps}
\end{figure}

While $\epsilon$ controls the strength of the signal present in the network,
$\delta$ controls the signal present in the features,
with larger values of $\delta$ allowing stronger delineation of the block of interest from the rest of the graph based on features alone.
To demonstrate this, we consider the same experiment as that summarized in Figure~\ref{fig:varyeps},
but this time fixing $\epsilon=0.25$ and varying $\delta = 0,0.5,1,1.5,2$.
The results are summarized in Figure~\ref{fig:varydel},
where we plot $\rGF(k)-\rG(k)$ over $k\in(1,10,20,30,40)$ where
$\rG(k)$ is the number of vertices of interest in $G_2$ nominated in the top $k$ by $\VNA$ applied to $(G_1,G_2)$ (i.e., ignoring vertex features).
As with Figure~\ref{fig:varyeps}, we see that as $\delta$ increases,
the gain in VN performance from incorporating vertex features increases.
For small values of $\delta$, features are slightly detrimental to performance,
again owing to the fact that there is insufficient signal present in them to
differentiate the vertices of interest from the rest of the network.

In each of Figures~\ref{fig:varyeps} and~\ref{fig:varydel}, using one of the two available data modalities (networks or features) gives performance that, while significantly better than chance, is suboptimal.
These experiments suggest that combining informative network structure with informative features should yield better VN performance than utilizing either source in isolation.

\begin{figure}[t!]
    \centering
    \includegraphics[width=0.7\textwidth]{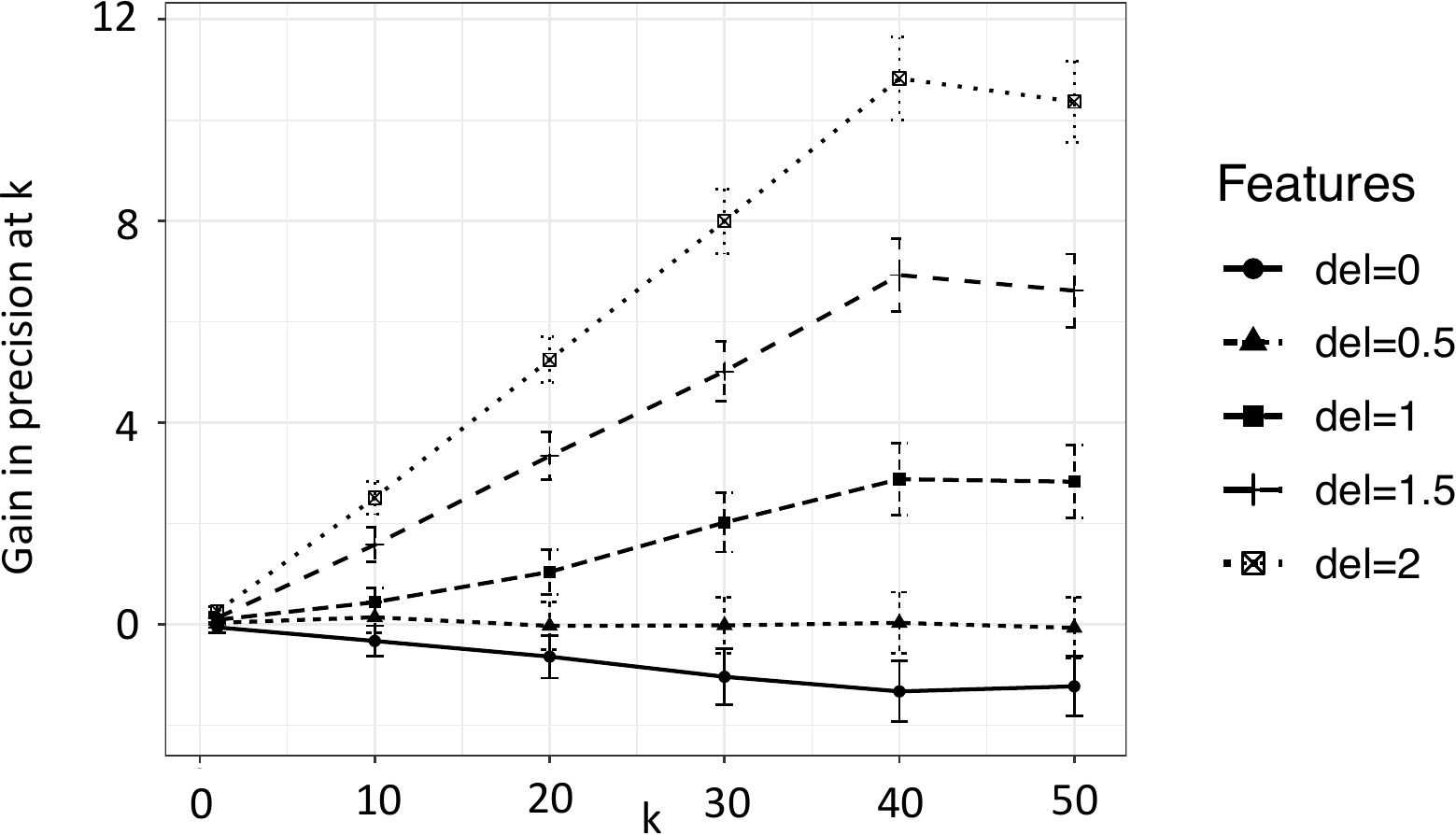} 
        \caption{Improvement in vertex nomination performance
        under the stochastic block model specified above,
        as a function of $\delta = 0,0.5,1,1.5,2$ for fixed $\epsilon=0.25$,
        based on $10$ randomly chosen seeded vertices.
        The plot shows $\rGF(k)-\rG(k)$, as defined in Equation~\eqref{eq:consis2},
        for $k\in \{ 1,10,20,30,40 \}$.
    Results are averaged over 100 Monte Carlo trials, with error bars indicating two standard errors of the mean.}
    \label{fig:varydel}
\end{figure}

\subsection{{\em C. Elegans}}
\label{sec:worms}
We next consider a real data example derived from the {\em C.\ elegans} connectome, as presented in \cite{celegans1,celegans2}.
In this data, vertices correspond to neurons in the {\em C.~elegans}, with edges encoding which pairs of neurons form synapses.
The data capture the connectivity among the 302 labeled neurons in the hermaphroditic {\em C.\ elegans} brain for two different synapse types called {\em electrical gap junctions}
and {\em chemical synapses}.
These two different synaptic types yield two distinct connectomes (i.e., brain networks)
capturing the two different kinds of interactions between neurons.
After preprocessing the data,
including removing neurons that are isolates in either connectome, symmetrizing the directed chemical connectome and removing self-loops (see \cite{chen2016joint} for details),
we obtain two weighted networks on $253$ shared vertices: $G_c$, capturing the chemical synapses, and $G_e$, capturing the electrical gap junction synapses.
The graphs are further endowed with vertex labels (i.e., vertex features), which assign each vertex (i.e., neuron) to one of three neuronal types: sensory, motor, or inter-neurons.
\begin{figure}[t!]
    \centering
\includegraphics[width=1\textwidth]{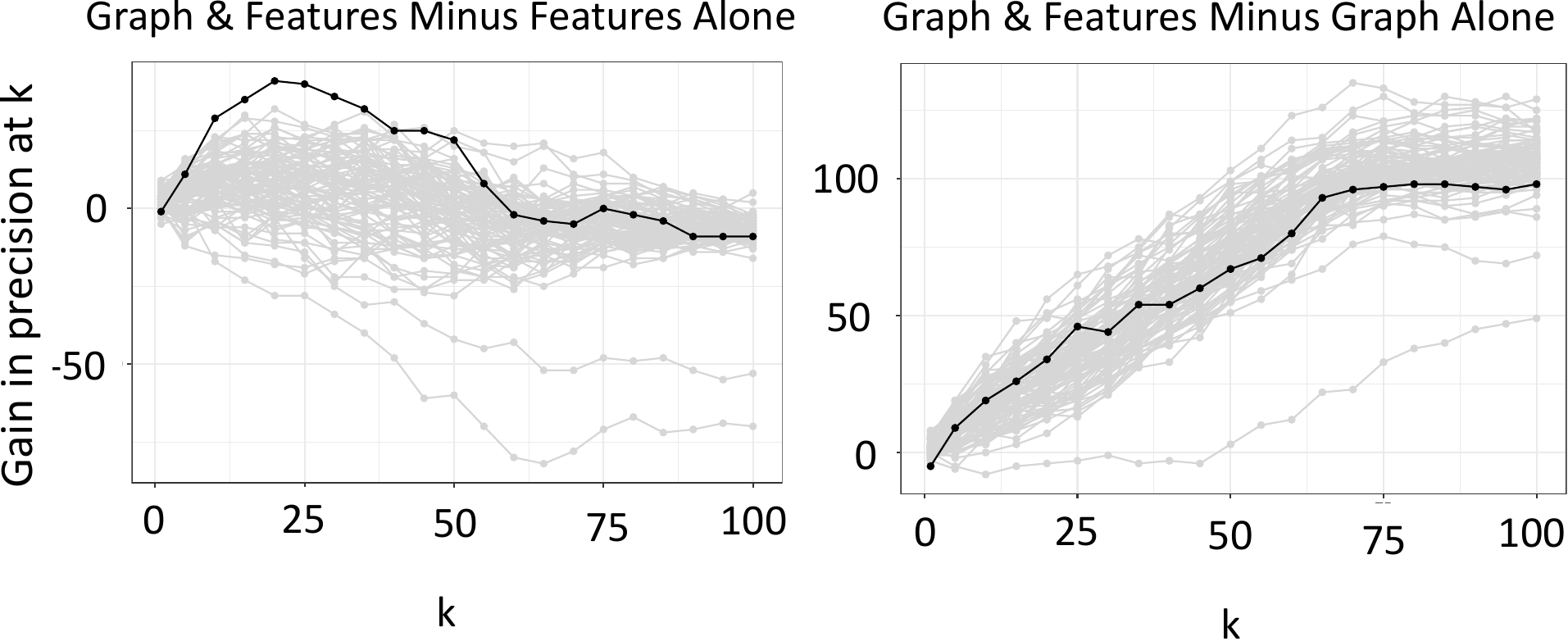}
        \caption{Improvement in vertex nomination performance when using both
            network structure and neuronal type features, compared to
            (a) using network structure only and
            (b) using neuronal type features only.
            Performance was measured according to the number of vertices of interest whose corresponding match was ranked in the top $k$ (i.e., $|\{v\in V(G_c)\in\text{ s.t. }\rank_{\Phi}(v')\leq k\}|$) as a function of $k$.
            Each grey line corresponds to a single trial,
            and shows the improvement of this performance measure when using both
            network structure and vertex features as compared to the performance of its feature-oblivious (left) or network-oblivious (right) counterpart.
            We have highlighted in black a single ``good'' trial in each subplot.}
    \label{fig:worm}
\end{figure}

Each of the 253 neurons in $G_c$ has a known true corresponding neuron in $G_e$.
Thus, there is a sensible ground truth in a vertex nomination problem
across $G_c$ and $G_e$, in the sense that each vertex in $G_c$ has one and only one corresponding vertex in $G_e$.
As such, this data provides a natural setting for evaluating vertex nomination performance.
We thus consider the following experiment:
a vertex $v$ in $G_c$ is chosen uniformly at random and designated as the vertex of interest.
An additional 20 vertices are sampled to serve as seeded vertices for the Procrustes alignment step, and the $\vnase$ nomination scheme is applied as outlined previously.
Performance was measured by computing the number of vertices of interest whose corresponding match was ranked in the top $k$, according to
\begin{equation*}
r(k)= \left| \{v \in V(G_c)\in\text{ s.t. } \rank_{\Phi}(v')\leq k\} \right|,~~~k=1,2,\dots,253,
\end{equation*}
where $v'$ denotes the vertex in $G_e$ corresponding to $v \in V(G_c)$.
We denote by $\rG, \rF$ and $\rGF$,
the performance of VN applied to, respectively,
the network only; the features only; and both the network and features jointly.
Figure~\ref{fig:worm} summarizes the result of 100 independent Monte Carlo trials of this experiment.
Each curve in the figure corresponds to one trial.
In each trial, we compared the performance of VN based on both network structure and vertex features against using either only network structure (i.e., feature-oblivious VN) or only vertex features (i.e., network-oblivious VN).
The left panel of Figure~\ref{fig:worm} shows VN performance based on both network structure and neuronal type features, which we append onto $\widehat X_c$ and $\widehat{Y}_e$ in step (iii) above,
minus performance of the scheme using the graph alone (i.e., $\rGF-\rG$).
Similarly, in the right panel of Figure~\ref{fig:worm},
we consider VN performance based on the graph with the neuronal type features minus performance of the scheme in the setting with only neuronal features (i.e., $\rGF-\rF$).
Within each plot, we have selected one line (i.e., one trial) to highlight, corresponding to a trial with a comparatively ``good'' seed set.
Note that same trial (and thus the same seed set) is highlighted in both panels.
Performance is also summarized in Table~\ref{tab:worm}.
\begin{table}[t!]
\centering
\begin{tabular}{|c|c|c|c|c|c|c|c|c|}
\hline
				& $k=1$ & $k=5$ & $k=10$ & $k=15$ & $k=20$ & $k=25$ & $k=30$ & $k=50$\\
\hline
$\rGF-\rF$ 	& 1.73 & 4.76  &7.60 & 7.83 & 8.12  &7.25 & 6.14& -1.29\\
\hline 
$\rGF-\rG$ 	& 1.53 &  7.94 & 15.55 & 22.31  &28.30 & 34.42 & 41.35& 73.43 \\
\hline
\end{tabular}
\caption{Mean values of $\rGF-\rG$ and $\rGF-\rF$ over the range of values of $k$ considered in the experiment.  Note that the mean of $\rGF-\rF$ is less than 0 for larger $k$.}
\label{tab:worm}
\end{table}

Using only the neuronal features for vertex nomination amounts to considering a coarse clustering of the neurons into the three neuronal types.
As such, recovering correspondences across the two networks networks based only on this feature information is effectively at chance, conditioned on the neuronal type.
When $\rGF-\rF$ is approximately $0$, the graph is effectively providing only enough information to coarsely cluster the vertices into their neuronal types.
Examining Figure~\ref{fig:worm} and Table~\ref{tab:worm}, it is clear
that incorporating features adds significant signal compared to only considering network structure.
Indeed, $\rGF-\rG$ is uniformly positive.

Interestingly, here the right-hand panel suggests that
adding the network topology improves performance compared to a scheme that
only uses features.
Of course, in general we expect that network structure should add significant signal to the features, but this observation is surprising in the present setting.
In the present data set, it is known that the network topology differs dramatically across the two different synapse types.
For example, $G_c$ has more than three times the edges of $G_e$.
As a result, it is notoriously difficult to discover the vertex correspondence across this pair of networks using only topology.
Indeed, state-of-the-art network alignment algorithms only recover approximately $5\%$ of the correspondences correctly even using $50$ a priori known seeded vertices \cite{patsolic2014seeded}.
It is thus not immediate that there is sufficient signal in the networks to identify individual neurons across networks beyond their vertex type.
While the features add significant signal to the network, the graph also adds signal to the features.
For small $k$, which are typically most important for most vertex nomination problems, $\rGF-\rF$ is positive on average, and for well-chosen seed sets (see the black lines in the figure), this difference can be dramatic. 

\subsection{Wikipedia data}
\label{wiki}

As another illustration of vertex nomination with features on real data, we consider a pair of networks derived from Wikipedia articles in English and French.
We begin with a network $G_{\mathrm{EN}}$ whose vertices correspond to English language Wikipedia articles and whose edges join pairs of articles reachable one from another via a hyperlink.
The English language network corresponds to the $n =1382$ articles within two hops of the article titled ``Algebraic Geometry'', with the articles grouped into 6 ``types'' according to their broad topics; see \cite {ma2012fusion} for a detailed description.
We then consider a paired network $G_{\mathrm{FR}}$ of $n=1382$ vertices corresponding to French language Wikipedia articles on the same topics, with correspondence across these networks encoding whether or not one article's title is an exact or approximate translation of the other.
The hyperlink structure among the articles within each language yields a natural network structure, and the semantic content of the pages, as encoded via a bag-of-words representation, provides a natural choice of vertex features.
As in \cite{shen2017manifold}, we consider capturing both the network and semantic feature information within each network via dissimilarity matrices.
We use shortest path dissimilarity in the hyperlink graph and cosine dissimilarity between extracted text features.
This procedure yields four dissimilarity matrices, two for each of English and French Wikipedia.
We then embed the pages according to these dissimilarity measures using canonical multidimensional scaling; see \cite{borg2005modern} for detail.
This yields four embeddings, corresponding to each pairing of language (English or French) and structure (network or semantic features).

In order to disentangle the information contained in the network structure and the features, we consider the following experiment.
Using $s=10$ randomly chosen ``seeded'' vertices across the networks (recall that seeded vertices are those whose correspondence is known a priori across networks), we align embeddings using orthogonal Procrustes alignment \cite{GowDij2004}.
We then cluster the combined point clouds as in step (iv) of Section~\ref{sec:exp} and nominate across the datasets using step (v) of Section~\ref{sec:exp}.
We use $K=6$ clusters in \texttt{MClust} to reflect the six different broad article types in the Wikipedia data.
We next use the JOFC algorithm of \cite{priebe2013manifold,lyzinski2017fast} to jointly embed the English dissimilarities and (separately) jointly embed the French dissimilarities.
Within each language, we then average across the embedded dissimilarities, and repeat the above procedure: Procrustes alignment using $s=10$ randomly chosen seeded vertices across the networks, followed by clustering and nomination according to steps (iv) and (v) outlined at the beginning of Section~\ref{sec:exp}.
This procedure was repeated with $25$ Monte Carlo replicates.
We note here that while the embedding procedure differs from adjacency spectral embedding, the core nomination strategy post-embedding is unchanged from that presented at the start of Section~\ref{sec:exp}.

As in the {\em C.~elegans} example, we consider each of the $1382$ articles as vertices of interest separately.
For each vertex $v$ in $G_{\mathrm{EN}}$, we consider the rank of its true corresponding match $v'$ in $G_{\mathrm{FR}}$, and record how many vertices have their true corresponding matches ranked in the top $k$, for varying values of $k$.
Figure~\ref{fig:wiki} summarizes how including both network structure and vertex features improves performance, as measured by the fraction of vertices whose true matches are ranked in the top $k$.
In the left (respectively, right) panel of Figure~\ref{fig:wiki}, we plot the performance when nominating across the jointly embedded graphs as compared with the performance when nominating across only the embedded graph dissimilarities (respectively, text feature dissimilarities).
Each light-colored curve in the figure corresponds to one of the 25 Monte Carlo replicates, with the dark curve representing the average across all 25 replicates.
In both settings, we see the overall positive effect of using both network and features when nominating across languages.
However, from the figure it is clear that there is an asymmetry in information content across network and features, in that the network topology contributes markedly less to the total performance gain than the textual feature information.
Moreover, while the inclusion of text features is nearly uniformly helpful as compared with the network alone (left panel), a poor choice of seed vertices may result in a situation wherein network information actually impedes performance compared to using text-derived features alone.
This can be seen in the light blue curves {\em below} zero in the right-hand panel of Figure~\ref{fig:wiki}.

\begin{figure}[t!]
    \centering
    \includegraphics[width=1\textwidth]{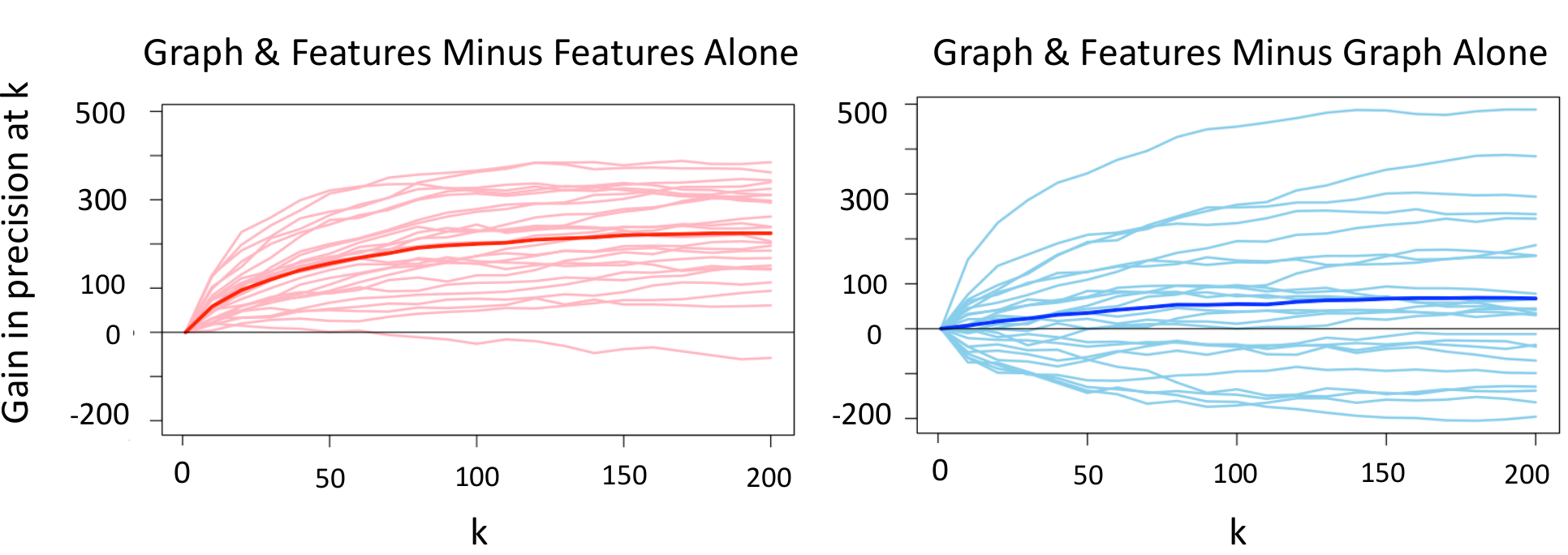} 
        \caption{Improvement in vertex nomination performance when using both network structure and semantic text features in the Wikipedia network example using (left) only network structure and (right) using semantic text features only.
            Performance was measured according to the number of vertices of interest whose corresponding match was ranked in the top $k$ (i.e., $|\{v\in V^*(G_{\mathrm{FR}}) \text{ s.t. }\rank_{\Phi}(v')\leq k\}|$) versus $k$.
            Each light-colored line corresponds to a single Monte Carlo replicate, with the dark curves representing the average across all replicates.}
    \label{fig:wiki}
\end{figure}

\section{Discussion}
\label{sec:discuss}
It is intuitively clear that informative features and network topology will together yield better performance in most network inference tasks compared to using either mode in isolation.
Indeed, in the context of vertex nomination, this has been established empirically across a host of application areas \cite{CopPri2012,marchette2011vertex}.
However, examples abound where the underlying network does not offer additional information for subsequent information retrieval, and may even be detrimental; see, for example,
\cite{rastogi2017vertex}.
In this paper, we have established the first (to our knowledge) theoretical exploration of the dual role of network and features, and we provide necessary and sufficient conditions under which VN performance can be improved by incorporating both network structure and features.
Along the way, we have formulated a framework for vertex nomination in richly-featured networks, and derived the analogue of Bayes optimality in this framework.
We view this work as constituting an initial step towards a more comprehensive understanding of the benefits of incorporating features into network data and complementing classical data with network structure.
A core goal of future work is to extend the framework presented here to incorporate continuous features; establish theoretical results supporting our empirical findings of the utility of features and network in the $\vnase$ algorithm; understand the role of missing or noisily observed features; and develop a framework for adversarial attack analysis in this richly-featured setting akin to that in \cite{agterberg2019vertex}.
\vspace{2mm}

\noindent \textbf{Acknowledgments} This material is based on research sponsored by the Air Force Research Laboratory and DARPA under agreement numbers FA8750-18-2-0035 and FA8750-20-2-1001,
and by NSF grants DMS-1646108 and DMS-2052918. This work is also supported in part by the D3M program of the Defense Advanced Research Projects Agency. The U.S. Government is authorized to reproduce and distribute reprints for Governmental purposes notwithstanding any copyright notation thereon. The views and conclusions contained herein are those of the authors and should not be interpreted as necessarily representing the official policies or endorsements, either expressed or implied, of the Air Force Research Laboratory and DARPA or the U.S. Government.  The authors also gratefully acknowledge the support of NIH grant BRAIN U01-NS108637. KL acknowledges the support of the University of Wisconsin-Madison Office of the Vice Chancellor for Research and Graduate Education with funding from the Wisconsin Alumni Research Foundation.

\appendix

\section{Algorithmic primitives}
Here, we provide background information and technical details related to the algorithmic primitives involved in the $\vnase$ scheme described in Section~\ref{sec:exp}.

\subsection{Passing to ranks and diagonal augmentation}
\label{sec:ptrda}
Consider a weighted adjacency matrix $A\in\mathbb{R}^{n\times n}$, and let $w \in \R^{s}$ be the vector of edge weights of $A$. Note that we are agnostic to the dimension of $w$, which will
vary according to whether $A$ is symmetric, hollow, etc. Define $r \in \R^s$ by taking
$r_i$ to be the rank of $w_i$ in the weight vector $w$, with ties broken by averaging ranks.
By the {\em pass-to-ranks} operation, we mean to replace the edge weights in $w$ with the
vector $2r/(s+1)$. That is, replacing the weighted edges of $A$ by their ranks.
Note that if $A$ is binary, the pass-to-ranks operation simply returns $A$ unchanged.

By {\em diagonal augmentation} we mean setting 
$$A_{i,i} =\sum_{j\neq i} A_{i,j}/(n-1)$$
for each $i=1,2,\cdots,n$.
In experiments, we find that these preprocessing steps are essential for robust and reliable
performance on real network data \cite{tang2018connectome}. 

\subsection{Adjacency Spectral Embedding}
\label{sec:ASE}

Given an undirected network with adjacency matrix $A \in \Rb^{n \times n}$, the $d$-dimensional Adjacency Spectral Embedding (ASE)
of $A$ yields a mapping of the $n$ vertices in the network to points in $d$-dimensional space in such a way that vertices that play similar structural roles
in the network are mapped to nearby points in $\Rb^d$ \cite{SusTanFisPri2012}.
\begin{definition} [Adjacency spectral embedding]\label{def:ASE}
	Given $d \in\mathbb{Z}_{>0}$, the {\em adjacency spectral embedding} (ASE) of $A$
	into $\Rb^{d}$ is defined by $\widehat{{ X}}={ U}_{{A}}
	{ S}_{{A}}^{1/2} \in \Rb^{n \times d}$ where
	\begin{equation*}|{ A}|=[{ U}_{{A}}|{ U}^{\perp}_{{A}}][{
		S}_{{A}} \oplus { S}^{\perp}_{{A}}][{
		U}_{{A}}|{ U}^{\perp}_{{A}}] \end{equation*}
	is the spectral decomposition of $|{A}| = ({A}^{T}{A})^{1/2}$, 
	${S}_{{A}}\in\mathbb{R}^{d\times d}$ is the diagonal matrix with the $d$ largest eigenvalues
	of $|{A}|$ on its diagonal and ${U}_{{A}}\in\mathbb{R}^{n\times d}$  has 
	columns which are the eigenvectors corresponding to the eigenvalues of ${S}_{{A}}$.
	The $i$-th row of $\widehat{X}$ corresponds to the position
	in $d$-dimensional Euclidean space to which the $i$-th vertex is mapped.
\end{definition}

\subsection{Seeds}
\label{sec:seeds}
In vertex nomination, vertices in the core $C$ are shared across the two networks, although the correspondence between $C\cap V_1$ and $C\cap V_2$ is unknown owing to the obfuscating function.
In many applications, however, some of these correspondences may be known ahead of time.
We refer to vertices in $C$ for which this correspondence is known as {\em seeded vertices},
and denote them by $S \subseteq C$.
Said another way, seeded vertices are vertices in $C$ whose labels are not obfuscated.
In this case, the obfuscating function would take the form $\mo_S:V_2\mapsto S\cup H$ where
\begin{equation*}
\mo_S(u)=\begin{cases}
u&\text{ if }u\in S\\
h\in H&\text{ if }u\in V_2\setminus S
\end{cases}
\end{equation*}
and $H$ is an obfuscating set of order $m-|S|$ satisfying $H\cap V_i=\emptyset$ for $i=1,2$.
Seeded vertices, and the information they provide, have proven to be valuable resources across both VN (e.g., 
\cite{FisLyzPaoChePri2015,lyzinski2016consistency,patsolic2017vertex})
and other network-alignment tasks (e.g., \cite{FAP,lyzinski2014seeded,mossel2019seeded}).

\subsection{Orthogonal Procrustes}
\label{sec:proc}
The $d$-dimensional adjacency spectral embedding of a network on $n$ vertices yields a collection
of $n$ points in $\Rb^d$, one point for each vertex.
A natural way to compare two networks on $n$ vertices is to compare the point clouds produced
by their adjacency spectral embeddings; see, e.g.,\cite{MT2}.
Approaches of this sort are especially natural in low-rank models, such as the
random dot product graph \cite{RDPGsurvey,RubCapTanPri2017}
and the stochastic block model.
In such models,  we can write the expectation of the adjacency matrix as
$\mathbb{E} A = X X^T$ for $X \in \Rb^{n \times d}$,
and the adjacency spectral embedding of $A$ is a natural estimate of $X$,
up to orthogonal rotation. That is, for some unknown orthogonal $Q \in \Rb^{d \times d}$,
$X$ and $\widehat{{ X}} Q$ are close.
Non-identifiabilities of this sort are inherent to latent space network models,
whereby transformations that preserve pairwise similarity of the latent positions
lead to identical distributions over networks \cite{ShaAst2017}.
Owing to this non-identifiability, comparison of two networks via
their adjacency spectral embeddings $\widehat{{ X}}$ and $\widehat{{ Y}}$ 
requires accounting for this unknown rotation.

Given matrices $X,Y\in\mathbb{R}^{n\times d}$,
the orthogonal Procrustes problem
seeks the orthogonal matrix $Q\in\Rb^{d \times d}$ that minimizes
$\|XQ-Y\|_F$ (where $\|\cdot\|_F$ is the Frobenius norm).
The problem is solved by computing the singular value decomposition
$X^T Y=U\Sigma V^T$, with the optimal $Q$ given then by $Q^*=UV^T$
\cite{schonemann1966generalized}.
We note that the orthogonal Procrustes problem is just one of a number of related
alignment problems for point clouds \cite{GowDij2004}.

\section{Proofs and supporting results}
Below we provide proofs of our main theoretical results and supporting lemmas.

\subsection{Proof of Theorem \ref{thm:BO}}
\label{sec:proof!}
Recall that $\mathcal{S}$ is the set of indices $i$ such that $\bbg_1^{(i)}$ and $\bbg_2^{(i)}$ are asymmetric as richly-featured networks (i.e., for $j=1,2$ there are no non-trivial f-automorphisms of $\bbg_j^{(i)}$).

To compare the VN loss of $\Phi^*$ to that of an arbitrary VN scheme $\Phi$, we will proceed as follows.
Let $k\leq m-1$ be fixed.
With
\begin{equation*}
    ({\bf G_1},{\bf G_2})= \left( (G_1,\bX,\bW),(G_2,\bY,\bZ) \right) \sim \FFt,
\end{equation*}
define
$A^j_v:=\{ \rank_{\Phi({\bf G_1},\mo({\bf G_2}),V^*)}(\mo(v)) = j\}$
for each $j\in[k]$.
Then we have that
\begin{align*}
\p( A^j_v)
&=\sum_{i\in\mathcal{S}}\p\left[ A^j_v\Big|(\bbg^{(i)}_1,[\mo(\bbg_2^{(i)})])\right] \p\left((\bbg^{(i)}_1,[\mo(\bbg_2^{(i)})])\right).
\end{align*}
Next, note that for each $v\in V^*$ and $i\in \mathcal{S}$,
\begin{align*}
&\left\{(\bbg_1,\bbg_2)\in 
\big( (\bbg_1^{(i)},[\mo(\bbg_2^{(i)})]\big)
: \rank_{\Phi(\bbg_1,\mo(\bbg_2),V^*)}(\mo(v)) =j\right\}\\
&=\left\{(\bbg_1,\bbg_2)\in 
\big( (\bbg_1^{(i)},[\mo(\bbg_2^{(i)})]\big)
: \Phi(\bbg_1,\mo(\bbg_2),V^*)[j]=\mo(v) \right\}\\
&=\Big\{(\bbg_1,\bbg_2)\in 
\big( (\bbg_1^{(i)},[\mo(\bbg_2^{(i)})]\big)
: \exists\text{ f-isomorphism }\sigma\text{ s.t. }\sigma(\mo(\bbg_2^{(i)}))=\mo(\bbg_2)\\
&\hspace{15mm}\text{ and }\sigma\big(\,
\Phi(\bbg_1^{(i)},\mo(\bbg_2^{(i)}),V^*)[j]\,
\big)= \mo(v)\Big\}\\
&=\big( \bbg_1^{(i)},[\mo(\bbg_2^{(i)})]\big)_
{\Phi(\bbg_1^{(i)},\mo(\bbg_2^{(i)}),V^*)[j]=\mo(v)}.
\end{align*}

To ease notation in what follows, we define the following key term for the support of $\FFt$ satisfying Assumption \ref{ass:A1}; i.e., on all  $(\bbg_1,\bbg_2)\in\bigcup_{i\in \mathcal{S}} \big( \bbg_1^{(i)},[\mo(\bbg_2^{(i)})]\big)$,
\begin{align*}
R_k(\Phi,\bbg_1,\bbg_2,V^*):&=\sum_{j\leq k}\sum_{v\in V^*}\p\left[ A^j_v\,\big|\,
\big( \bbg_1,\left[\mo(\bbg_2) \right]\big)\,\right]\\
&=\sum_{j\leq k}\sum_{v\in V^*}
\p\left[
\big( \bbg_1,\left[\mo(\bbg_2) \right]\big)_
{\Phi(\bbg_1,\mo(\bbg_2),V^*)[j]=\mo(v)}
\,\big|\,
\big( \bbg_1,\left[\mo(\bbg_2) \right]\big)\right]\\
&=\sum_{j\leq k}
\p\left[ 
\big( \bbg_1,\left[\mo(\bbg_2) \right]\big)_
{\Phi(\bbg_1,\mo(\bbg_2),V^*)[j]\in\mo(V^*)}
\,\big|\,\big( \bbg_1,\left[\mo(\bbg_2) \right]\big)\right],
\end{align*}
and note that, by definition of $\Phi^*$ as the optimal nomination scheme,
for any $i\in\mathcal{S}$, 
\begin{equation*}
R_k(\Phi,\bbg_1^{(i)},\bbg_2^{(i)},V^*)\leq 
R_k(\Phi^*,\bbg_1^{(i)},\bbg_2^{(i)},V^*).
\end{equation*}
Thus, for any FA-VN scheme $\Phi$, we have
\begin{equation*} \begin{aligned}
1-L_k(\Phi,V^*)&=\frac{1}{k}\sum_{v\in V^*}\p( \rank_{\Phi({\bf G_1},\mo({\bf G_2}),V^*)}(\mo(v)) \leq k) =\frac{1}{k}\sum_{j\leq k}\sum_{v\in V^*}\p( A_v^j) \\
&=\frac{1}{k}\sum_{j\leq k}\sum_{v\in V^*}
\sum_{i\in\mathcal{S}}\p\left( A^j_v\,\big|\,( \bbg_1^{(i)},[\mo(\bbg_2^{(i)})])\right)
 \p\left((\bbg^{(i)}_1,[\mo(\bbg_2^{(i)})])\right) \\
&=\frac{1}{k}\sum_{i\in\mathcal{S}}
R_k(\Phi,\bbg_1^{(i)},\bbg_2^{(i)},V^*) \p\left((\bbg^{(i)}_1,[\mo(\bbg_2^{(i)})])\right) \\
&\leq
\frac{1}{k}\sum_{i\in\mathcal{S}}
R_k(\Phi^*,\bbg_1^{(i)},\bbg_2^{(i)},V^*) \p\left((\bbg^{(i)}_1,[\mo(\bbg_2^{(i)})])\right)=1-L_k(\Phi^*,V^*),
\end{aligned} \end{equation*}
from which we deduce that $L_k(\Phi^*,V^*)\leq L_k(\Phi,V^*)$, completing the proof.

\subsection{ Proof of Theorem \ref{thm:feat}}
\label{sec:featpf}
Suppose that $\mathbb{I}(X_{\Phi^*}^k; (G_1,G_2) )=\mathbb{H}(X_{\Phi^*}^k)$,
whence $\bH(X_{\Phi^*}^k| (G_1,G_2))=0$
and thus for each $(g_1,g_2)$ with $\p((G_1,G_2)=(g_1,g_2))>0$
it holds for all $\xi\in \mathcal{T}_H^k$ that
\begin{equation*}
    \p(X_{\Phi^*}^k=\xi\,|\,(G_1,G_2)=(g_1,g_2))\in\{0,1\}.
\end{equation*}
For each $(g_1,g_2)$, let $\xi_{g_1,g_2}$ denote the unique element in the support of $X_{\Phi^*}^k\,|\,(G_1,G_2)=(g_1,g_2)$.
With this notation in hand, we define the FO-VN scheme $\Psi$ as follows.
For $\bbg_1=(g_1,\bx,\bw)$ and $\bbg_2=(g_2,\by,\bz)$, take
\begin{equation*}
    \Psi(\bbg_1,\mo(\bbg_2),V^*)=\widehat{\xi}_{g_1,g_2},
\end{equation*}
where $\widehat{\xi}_{g_1,g_2}\in \mathcal{T}_H$ satisfies
\begin{itemize}
\item[i.] $\widehat{\xi}_{g_1,g_2}[1:k]=\xi_{g_1,g_2}$;
\item[ii.] $\widehat{\xi}_{g_1,g_2}[k+1:m]$ is ordered lexicographically according to some predefined total ordering of $H$.
\end{itemize}
Then $\Psi$ is an FO-VN scheme by construction, and 
\begin{equation*}
    \Psi({\bf G}_1,\mo({\bf G}_2),V^*)[1:k] = 
	\Phi^*({\bf G}_1,\mo({\bf G}_2),V^*)[1:k]
	~\text{ almost surely,}
\end{equation*}
from which $L_k(\Phi^*,V^*)= L_k(\Psi,V^*)\geq L_k(\Psi^*,V^*)\geq L_k(\Phi^*,V^*)$
and it follows that $L_k(\Phi^*,V^*)=L_k(\Psi^*,V^*)$, as desired.

To prove the other half of the Theorem, we 
proceed as follows.
The assumption that 
$L_k(\Phi^*,V^*)=L_k(\Psi^*,V^*)$ implies that (with notation as in Section~\ref{sec:proof!}),
\begin{align*}
0&=L_k(\Phi^*,V^*)-L_k(\Psi^*,V^*)\\
&=\frac{1}{k}\sum_{i\in\mathcal{S}}\left[
R_k(\Psi^*,\bbg_1^{(i)},\bbg_2^{(i)},V^*)-
R_k(\Phi^*,\bbg_1^{(i)},\bbg_2^{(i)},V^*)\right] \p\left((\bbg_1^{(i)},[\mo(\bbg_2^{(i)})])\right),
\end{align*}
and therefore, since
\begin{equation*}
R_k(\Psi^*,\bbg_1^{(i)},\bbg_2^{(i)},V^*) \le
R_k(\Phi^*,\bbg_1^{(i)},\bbg_2^{(i)},V^*)
\end{equation*}
for all $i\in\mathcal{S}$, we conclude that
\begin{equation*}
R_k(\Psi^*,\bbg_1^{(i)},\bbg_2^{(i)},V^*)=
R_k(\Phi^*,\bbg_1^{(i)},\bbg_2^{(i)},V^*).
\end{equation*}
Therefore, there exists a tie-breaking scheme in the definition of $\Phi^*$ that yields
\begin{equation*}
\Psi^*(\bbg_1^{(i)},\mo(\bbg_2^{(i)}),V^*)[1:k]{=}\Phi^*(\bbg_1^{(i)},\mo(\bbg_2^{(i)}),V^*)[1:k]
\end{equation*}
for all $i\in\mathcal{S}$, and hence
\begin{equation*}
\Psi^*({\bf G}_1,\mo({\bf G}_2),V^*)[1:k]\stackrel{a.s.}{=}\Phi^*({\bf G}_1,\mo({\bf G}_2),V^*)[1:k].
\end{equation*}
We therefore have that 
$\bH(X_{\Phi^*}^k| (G_1,G_2))=\bH(X_{\Psi^*}^k| (G_1,G_2))$.
Since $\Psi^*$ is a constant given $(G_1,G_2)$, we have
$\bH(X_{\Psi^*}^k| (G_1,G_2))=0$, and therefore $\bH(X_{\Phi^*}^k| (G_1,G_2))=0$,
completing the proof.

\subsection{Proof of Theorem \ref{thm:topology}}
\label{sec:top_proof}
We assume throughout that $F$ satisfies Assumption \ref{ass:A3}.
Suppose $\mathbb{I}(X_{\Phi^*}^k; (\bX,\bY) )=\mathbb{H}(X_{\Phi^*}^k)$,
implying that $\bH(X_{\Phi^*}^k| (\bX,\bY))=0$.
For each $(\bx,\by)$ satisfying 
$\p((\bX,\bY)=(\bx,\by))>0$ and each $\xi\in \mathcal{T}_H^k$, we have
\begin{equation*}
    \p(X_{\Phi^*}^k=\xi\,|\,(\bX,\bY)=(\bx,\by))\in\{0,1\}.
\end{equation*}
For each $(\bx,\by)$, let $\xi_{\bx,\by}$ be the unique element in the support of $X_{\Phi^*}^k\,|\,(\bX,\bY)=(\bx,\by)$.
We define the NO-VN scheme $\Xi$ as follows.
For $\bbg_1=(g_1,\bx,\bw)$ and $\bbg_2=(g_2,\by,\bz)$, we take
\begin{equation*}
    \Xi(\bbg_1,\mo(\bbg_2),V^*)=\widehat{\xi}_{\bx,\by},
\end{equation*}
where $\widehat{\xi}_{\bx,\by}\in \mathcal{T}_H$ satisfies
\begin{itemize}
\item[i.] $\widehat{\xi}_{\bx,\by}[1:k]=\xi_{\bx,\by}$;
\item[ii.] $\widehat{\xi}_{\bx,\by}[k+1:m]$ is ordered lexicographically according to some predefined total ordering of $H$.
\end{itemize}
$\Xi$ is an NO-VN scheme by construction, and 
\begin{equation*}
    \Xi({\bf G}_1,\mo({\bf G}_2),V^*)[1:k]\stackrel{a.s.}{=}\Phi^*({\bf G}_1,\mo({\bf G}_2),V^*)[1:k],
\end{equation*}
from which
$L_k(\Phi^*,V^*)= L_k(\Xi,V^*)\geq L_k(\Xi^*,V^*)\geq L_k(\Phi^*,V^*)$,
and we conclude that $L_k(\Phi^*,V^*)=L_k(\Xi^*,V^*)$.

To prove the other half of the theorem, we
note that 
$L_k(\Phi^*,V^*)=L_k(\Xi^*,V^*)$ implies that (with notation as in Section \ref{sec:proof!}),
\begin{align*}
0&=L_k(\Phi^*,V^*)-L_k(\Xi^*,V^*)\\
&=\frac{1}{k}\sum_{i\in\mathcal{S}}\left[
R_k(\Xi^*,\bbg_1^{(i)},\bbg_2^{(i)},V^*)-
R_k(\Phi^*,\bbg_1^{(i)},\bbg_2^{(i)},V^*)\right] \p\left((\bbg_1^{(i)},[\mo(\bbg_2^{(i)})])\right),
\end{align*}
since
\begin{equation*}
R_k(\Xi^*,\bbg_1^{(i)},\bbg_2^{(i)},V^*) \le
R_k(\Phi^*,\bbg_1^{(i)},\bbg_2^{(i)},V^*)
\end{equation*}
for all $i\in\mathcal{S}$, we conclude that
\begin{equation*}
R_k(\Xi^*,\bbg_1^{(i)},\bbg_2^{(i)},V^*)=
R_k(\Phi^*,\bbg_1^{(i)},\bbg_2^{(i)},V^*).
\end{equation*}

Thus, there exists a tie-breaking scheme in the definition of $\Phi^*$
such that
\begin{equation*}
    \Xi^*(\bbg_1^{(i)},\mo(\bbg_2^{(i)}),V^*)[1:k]
    {=}
    \Phi^*(\bbg_1^{(i)},\mo(\bbg_2^{(i)}),V^*)[1:k]
\end{equation*}
for all $i\in\mathcal{S}$, and hence
\begin{equation*}
    \Xi^*({\bf G}_1,\mo({\bf G}_2),V^*)[1:k]\stackrel{a.s.}{=}\Phi^*({\bf G}_1,\mo({\bf G}_2),V^*)[1:k].
\end{equation*}
We then have that
$\bH(X_{\Phi^*}^k| (\bX,\bY))=\bH(X_{\Xi^*}^k| (\bX,\bY))$.
Since $\Xi^*$ is a constant given $(\bX,\bY)$, we have
$\bH(X_{\Xi^*}^k| (\bX,\bY))=0$, whence we conclude that
$\bH(X_{\Phi^*}^k| (\bX,\bY))=0$, which completes the proof.

\subsection{Supporting lemmas}
The following lemma follows from our assumption of asymmetry.
\begin{lemma}
If 
$(\bbg_1,\bbg_2)\in 
\big( \bbg_1^{(i)},[\mo(\bbg_2^{(i)})]\big)$ for $i\in\mathcal{S}$, then
\vspace{-1mm}
\begin{equation*}
\big( \bbg_1,\left[\mo(\bbg_2) \right]\big)_
{\Phi(\bbg_1,\mo(\bbg_2),V^*)[j]=\mo(v)}=
\big( \bbg_1^{(i)},[\mo(\bbg_2^{(i)})]\big)_
{\Phi(\bbg_1^{(i)},\mo(\bbg_2^{(i)}),V^*)[j]=\mo(v)}.
\end{equation*}
\end{lemma} 
\begin{proof}
By the assumption that $(\bbg_1,\bbg_2)\in 
\big( \bbg_1^{(i)},[\mo(\bbg_2^{(i)})]\big)$, we have that $\bbg_1=\bbg_1^{(i)}$, and there exists an isomorphism $\tau$ such that $\bbg_2=\tau(\bbg_2^{(i)})$.
From our assumption that $i\in\mathcal{S}$ and the consistency criteria in Definition \ref{def:VN}, 
\vspace{-1mm}
\begin{equation*}
\Phi(\bbg_1,\mo(\bbg_2),V^*)=\tau(\Phi(\bbg_1^{(i)},\mo(\bbg_2^{(i)}),V^*)).
\end{equation*}
A similar argument shows that
$\big( \bbg_1,\left[\mo(\bbg_2) \right]\big)=\big( \bbg_1^{(i)},[\mo(\bbg_2^{(i)})]\big)$.
We then have that 
\begin{equation*} \begin{aligned}
&\big( \bbg_1,\left[\mo(\bbg_2) \right]\big)_
{\Phi(\bbg_1,\mo(\bbg_2),V^*)[j]=\mo(v)}\\
&~~~~~~=\left\{(\bbg_1',\bbg_2')\in 
\big( \bbg_1,\left[\mo(\bbg_2) \right]\big)
\text{ s.t. }\rank_{\Phi(\bbg_1',\mo(\bbg_2'),V^*)}(\mo(v)) =j\right\}\\
&~~~~~~=\left\{(\bbg_1',\bbg_2')\in 
\big( \bbg_1^{(i)},[\mo(\bbg_2^{(i)})]\big)
\text{ s.t. }\rank_{\Phi(\bbg_1',\mo(\bbg_2'),V^*)}(\mo(v)) =j\right\}\\
&~~~~~~=\big( \bbg_1^{(i)},[\mo(\bbg_2^{(i)})]\big)_
{\Phi(\bbg_1^{(i)},\mo(\bbg_2^{(i)}),V^*)[j]=\mo(v)},
\end{aligned} \end{equation*}
as we set out to show.
\end{proof}

\vspace{-5mm}

\begin{lemma}
\label{lem:R}
Let $\Phi^*$ be a Bayes optimal VN scheme, and let $\Phi$ be any other VN scheme.
For any $(\bbg_1,\bbg_2)\in\bigcup_{i\in \mathcal{S}} \big( \bbg_1^{(i)},[\mo(\bbg_2^{(i)})]\big)$,
\begin{align*}
R_k(\Phi,\bbg_1,\bbg_2,V^*) \leq R_k(\Phi^*,\bbg_1,\bbg_2,V^*).
\end{align*}
\end{lemma}
\begin{proof}
If there exists an $i\in\mathcal{S}$ such that $(\bbg_1,\bbg_2)=(\bbg_1^{(i)},\bbg_2^{(i)})$, the result follows from the definition of $\Phi^*$.
Consider then
\begin{equation*}
    (\bbg_1,\bbg_2)\in\left\{\bigcup_{i\in \mathcal{S}} \big( \bbg_1^{(i)},[\mo(\bbg_2^{(i)})]\big)\right\}
    \bigg \backslash
    \left\{\{(\bbg_1^{(i)},\bbg_2^{(i)})\}_{i\in\mathcal{S}} \right\},
\end{equation*}
and let 
$i' \in \mathcal{S}$ be such that $(\bbg_1,\bbg_2)\in \big( \bbg_1^{(i')},[\mo(\bbg_2^{(i')})]\big)$.
We have that
\begin{align*}
R_k(\Phi,\bbg_1,\bbg_2,V^*)&=\sum_{j\leq k}
\p\left( 
\big( \bbg_1,\left[\mo(\bbg_2) \right]\big)_
{\Phi(\bbg_1,\mo(\bbg_2),V^*)[j]\in\mo(V^*)}
\,\big|\,\big( \bbg_1,\left[\mo(\bbg_2) \right]\big)\right)\\
&=\sum_{j\leq k}
\p\left( 
\big( \bbg_1^{(i')},[\mo(\bbg_2^{(i')})]\big)_
{\Phi(\bbg_1^{(i')},\mo(\bbg_2^{(i')}),V^*)[j]\in\mo(V^*)}
\,\big|\,\big( \bbg_1^{(i')},[\mo(\bbg_2^{(i')})]\big)\right)\\
&\leq\sum_{j\leq k}
\p\left( 
\big( \bbg_1^{(i')},[\mo(\bbg_2^{(i')})]\big)_
{\Phi^*(\bbg_1^{(i')},\mo(\bbg_2^{(i')}),V^*)[j]\in\mo(V^*)}
\,\big|\,\big( \bbg_1^{(i')},[\mo(\bbg_2^{(i')})]\big)\right)\\
&=\sum_{j\leq k}
\p\left( 
\big( \bbg_1,\left[\mo(\bbg_2) \right]\big)_
{\Phi^*(\bbg_1,\mo(\bbg_2),V^*)[j]\in\mo(V^*)}
\,\big|\,\big( \bbg_1,\left[\mo(\bbg_2) \right]\big)\right)\\
&=R_k(\Phi^*,\bbg_1,\bbg_2,V^*),
\end{align*}
where the inequality follows from the optimality of $\Phi^*$.
\end{proof}

\vskip 0.2in
\bibliographystyle{plain} 
\bibliography{biblio}

\end{document}